\documentclass[12pt]{amsart}
\usepackage{latexsym,amssymb,amsmath,amsthm}
\usepackage{comment,graphicx,color,morefloats}
\usepackage[section]{placeins}
\newtheorem{thm}{Theorem}
\newtheorem{coro}[thm]{Corollary}
\newtheorem{lemma}[thm]{Lemma}
\newtheorem{propo}[thm]{Proposition}
\theoremstyle{definition}
\newtheorem{defn}{Definition}

\definecolor{red1}{rgb}{1,0.9,0.9}
\definecolor{blue1}{rgb}{0.9,0.9,1}
\definecolor{green1}{rgb}{0.9,1,0.9}
\definecolor{yellow1}{rgb}{1,1,0.9}
\definecolor{yellow2}{rgb}{1,1,0.8}

\title{A Sard theorem for graph theory}
\author{Oliver Knill}
\date{August 23, 2015}
\address{ Department of Mathematics \\ Harvard University \\ Cambridge, MA, 02138 }

\begin{document}
\maketitle

\begin{abstract}
The zero locus of a function $f$ on a graph $G$ is defined as the graph for which
the  vertex set consists of all complete subgraphs of $G$, on which $f$ changes sign and where $x,y$ 
are connected if one is contained in the other. For $d$-graphs, finite simple
graphs for which every unit sphere is a $d$-sphere,
the zero locus of $(f-c)$ is a $(d-1)$-graph for all $c$ different from the range of $f$.
If this Sard lemma is inductively applied to an ordered list functions $f_1,\dots,f_k$ 
in which the functions are extended on the level surfaces, the set
of critical values $(c_1,\dots,c_k)$ for which $F-c=0$ is not a $(d-k)$-graph is a 
finite set. This discrete Sard result allows to 
construct explicit graphs triangulating a given algebraic set. We also look at a second
setup: for a function $F$ from the vertex set to $R^k$, we give conditions for 
which the simultaneous discrete algebraic set $\{ F=c \}$ defined as the set of
simplices of dimension $\in \{k, k+1,\dots,n\}$ on which all $f_i$ change sign, 
is a $(d-k)$-graph in the barycentric refinement of $G$. This
maximal rank condition is adapted from the continuum and the graph $\{ F=c \}$ is a $(n-k)$-graph.
While now, the critical values can have positive measure, we are closer to calculus:
for $k=2$ for example, extrema of functions $f$ under a constraint $\{g=c\}$ happen at points, 
where the gradients of $f$ and $g$ are parallel $\nabla f = \lambda \nabla g$, the 
Lagrange equations on the discrete network.  As for 
an application, we illustrate eigenfunctions of $d$-graphs and especially the
second eigenvector of $3$-spheres, which by Courant-Fiedler has exactly 
two nodal regions. The separating nodal surface of the second
eigenfunction $f_1$ always appears to be a 2-sphere in experiments. 
By Jordan-Schoenfliess, both nodal regions would then be 3-balls and the double nodal curve 
$\{f_1=0,f_2=0\}$ would be an un-knotted curve in the 3-sphere. Graph theory allows 
to approach such unexplored concepts experimentally, as the corresponding question are open even
classically for nodal surfaces of the ground state of the Laplacian of a 
Riemannian $3$-sphere $M$. 
\end{abstract}

\section{Introduction}

We explore vector-valued functions $F$ on the vertex set of a finite simple graph $G =(V,E)$. Most of the notions introduced here 
are defined for general finite simple graphs. But as we are interested in Lagrange extremization,
Morse and  Sard type results in graph theory as well as
questions in the spectral theory of the Laplacian on graphs related to Laplacians of Riemannian manifolds, 
we often assume $G$ to be a $d$-graph, which is a finite simple graph, for which all unit spheres are $(d-1)$-spheres in the 
sense of Evako \cite{KnillJordan}. In a first setup, more suited for Sard,
for all except finitely many choices of $c \in R^d$, the graph $\{ F=c \} = \{ f_1=c_1, \dots , f_k = c_k \}$ 
is a $(d-k)$-graph, in a second setup, closer to classical calculus, we need to satisfy locally a maximal
rank condition to assure that the graph representing the discrete algebraic set $\{F = c \}$ is a $(d-k)$-graph. \\

The first part of the story parallels classical calculus and deals with the concept of level surface. 
It starts with a pleasant surprise when looking at a single
function: there is a strong Sard regularity for a $d$-graph $G$: given a hyper-surface 
$G_f(x) = \{ f = c \}$ defined as the graph with vertex set consisting of all simplices on which $f-c$ changes sign, 
the graph $\{ f=c \}$ is a $(d-1)$-graph if $c$ is not a value taken by $f$. The topology of $G_f(c)$ changes only for 
parameter values $c$ contained in the finite set $f(V)$. This observation was obtained when studying coloring problems 
\cite{knillgraphcoloring,knillgraphcoloring2,KnillNitishinskaya}, 
as a locally injective function on the vertex set is the same than a vertex coloring.  Before, in \cite{indexexpectation,
indexformula}, we just looked at the edges, on which $f$ changes sign and then completed the graph artificially.
Now, we have this automatic. We hope to apply this to investigate nodal regions of eigenfunctions of the Laplacian of a graph,
where we believe the answers to be the same for compact $d$-dimensional Riemannian manifolds or finite $d$-graphs. 
With the context of level surfaces one has the opportunity to look at nodal surfaces of eigenfunctions, which are 
also known as Chladni surfaces bounding nodal regions. \\

With more than one function, the situation changes as the singular set typically becomes larger in the discrete.
This is where the story splits. In the commutative setup, where we look at the zero locus of all functions simultaneously, 
Sard fails, while in the setup, where an ordered set of functions is considered, Sard will be true:
as in the classical Sard theorem \cite{Morse1939,Sard1942}, the set of critical values has zero measure.
The difference already is apparent if we take two random functions on a discrete $3$-sphere for
example. Then simultaneous level set $\{ f=0, g=0 \; \}$ is a graph without triangles but it is
rarely a $1$-graph, a finite union of circular graphs. The reason is that the tangent space is
a finite set and the probability of having two parallel gradients at a vertex does not have zero probability. 
However, if we look at $g=0$ on the two dimensional surface $f=0$, we get a finite union of cyclic graphs. 
In general, we salvage regularity and Sard by defining the algebraic set $\{ F=c \}$ 
in a different way by recursively building hyper surfaces:
start with the hyper surface $\{ f_1=c_1 \}$, then extend $f_2$ to the new graph and look at $\{ f_2=c_2 \}$ 
inside $\{f_1 =c_1 \}$. Now a vector value $c \in R^k$ is always regular if $c$ is not in the image of $F$ applied to 
recursively defined graph. This is the discrete analog of the multi-dimensional Sard theorem in classical analysis. 
The order in which the functions are taken, matters in the discrete. But this is not a surprise as we refine in each step the graph
and therefore have to extend the functions to the barycentric refinements. Its only for sufficiently smooth functions like 
eigenfunctions of low energy eigenvalues of the graph that the answer can be expected to be independent of the ordering. \\

The graph theoretical approach is useful for making experiments: take a 3-graph $G=(V,E)$ for example
and take two real valued functions $f,g$ on the vertex set $V$. If $c$ is not in the image of $f$, we can look at the level surface $f=c$ and 
extend $g$ to a function there (vertices are now simplices and we just average the function value $g$ to extend
the function). If $d$ is not in the image of $g$, then the $1$-graph $Z=\{ g=d \}$ is a subgraph of $\{ f=c \}$.
It is a finite set of closed curves in $G_2$. In other words, each connected component is a knot in the 3-sphere.
Unlike in the continuum, we do not have to worry about cases, where the 
knot intersects or self intersects. The Sard theorem assures that this will never happen in the discrete. We only have to
assure that the value $c$ is not in $f(V)$ and $d$ is not in $g(V)$.  \\

\begin{figure}[ph]
\scalebox{0.50}{\includegraphics{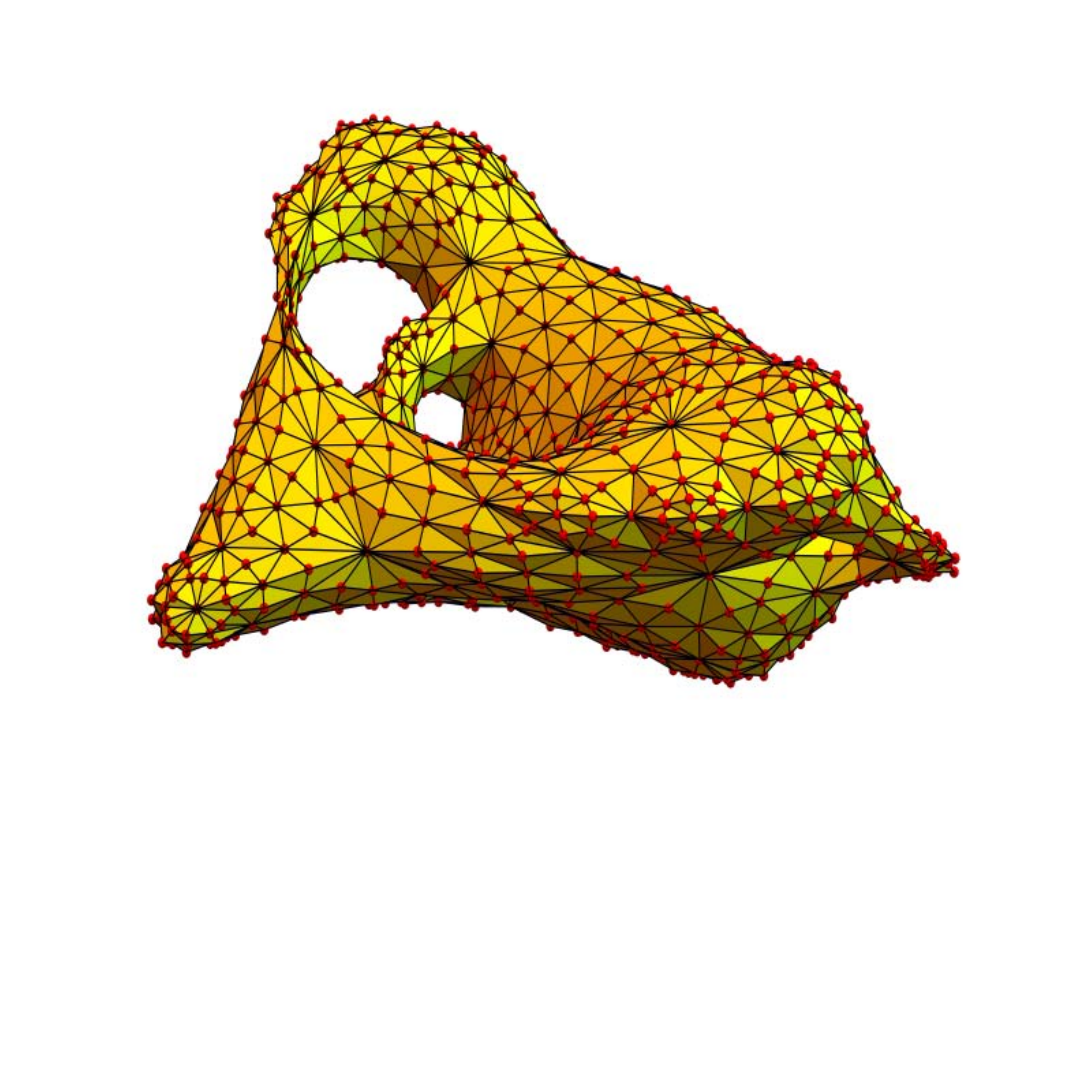}}
\caption{
We see an example of a level surface $G_f(c)=\{ f=c \}$ in a 3-sphere $G$.
The graph $G$ is a Barycentric refinement of a 3D octahedron with 8 vertices with
simplex cardinalities $\vec{v}=(80,464,768,384)$ satisfying $\chi(G)=\sum_i (-1)^i v_i=0$. The zero Euler characteristic
is a property shared by all 3-graphs like 3-manifolds. The graph $G$ admits an involution $T$ such that $G/T$ is a discrete projective 3-space. 
The algebraic level set $G_f(c)$ displayed here is a $2$-graph, a two-dimensional surface within the Barycentric refinement $G_1$ of $G$. 
The ambient sphere $G_1$ already has $1696$ vertices and $10912$ edges and $18432$ triangles, $9216$ tetrahedra.
The surface $\{ f = c \}$ is a discrete analogue of a $2$-dimensional projective variety. By the Sard lemma, 
the level sets $G_f(c) = \{ f=c \}$ are always nonsingular if $c \notin f(V)$. Furthermore  the 
surface $G_f(c)$ is $4$-colorable. }
\end{figure}

The second setup, which defines discrete algebraic sets $Z=\{ f_1=c_1,\dots ,f_{k}=c_{k} \}$ in the Barycentric refinement 
$G_1$ of $G$ requires a few definitions. We need conditions under which these sets are nonsingular in the sense that they again form 
a $(d-k)$-graph, graphs for which the unit sphere $S(x)$ is a sphere at every vertex. In general, the situation is 
the same as in the continuum, where varieties are not necessarily manifolds. 
The definition is straightforward: define $\{F=c \; \}$ as the set of simplices 
of dimension in $\{ k,\dots ,n \}$ in $G$ on which all functions $f_j$ change sign simultaneously. 
The graph $\{F = c\}$ is a subgraph of the Barycentric refinement $G_1 = G \times 1$ of $G$. 
Conditions for regularity could be formulated locally in terms of spheres in spheres. 
In some sense, a locally projective tangent bundle with discrete projective spaces associated to unit spheres replaces
the tangent bundle. It is still possible to define a vector space structure at each point and define a 
discrete gradient $\nabla f$ of a function. If $G$ is a $d$-graph, $x$ is a complete $K_{d+1}$ subgraph, then $\nabla f(x,x_0)$ 
is defined as the vector $\nabla f = \langle {\rm sign}(f(x_1)-f(x_0)), \dots,{\rm sign}(f(x_d)-f(x_0)) \rangle$. 
The local injectivity condition means that $\nabla f(x,x_0)$ is not zero for all $d$-simplices containing $x_0$.
This leads to the Sard Lemma~({\ref{sardlemma}}).
But now, when changing $c$, we also need to look at the values $c_i=f(x_i)$ taken by the function $f$ and investigate whether 
they are critical. What happens if $c$ passes a value $c_i = f(x_i)$? If $f<c_i$ is homotopic to $f \leq c_i$, then 
the set $\{ f = c \} \cap S(x)$ is a $(d-1)$-sphere which by
Jordan-Brower-Schoenflies \cite{KnillJordan} divides $S(x)$ into two complementary balls $S^-(x)$, $S^+(x)$.
The Poincar\'e-Hopf index $1-\chi(S^-(x)) = i_f(x)$ \cite{poincarehopf} is then zero. 
At a local minimum of $f$ for example, the graph $S^-(x)$ is empty so that the index is $1$. At a local maximum, $S^+(x)$ is empty and the 
index is either $1$ or $-1$ depending on whether the dimension of $G$ is even or odd. In two dimensions, where we look at discrete 
multi-variable calculus, maxima and minima have index $1$ and saddle points have index $-1$. \\

The set of hypersurfaces $M_j=\{ f =c_j \}$ form a contour map for which the individual leaves in general have different homotopy types. 
Topological transitions can happen only at values $f(V)$ of $f$ on the vertex set $V$ of the graph. 
We can relate the symmetric index $j_f(x)=[i_f(x)+i_{-f}(x)]/2$ at a vertex $x$ of $G$ with the $(d-2)$-dimensional 
graph $\{ f=c_k \}$ in the unit sphere $S(x)$. We have called this the graph $B_f(x)$ and showed that
the graph can be completed to become a $d$-graph. This completion can be done now more elegantly.
As pointed out in \cite{eveneuler}, for $d=4$ for example, we get for every 
locally injective function $f$ a $2$-dimensional surface $B_f$, the disjoint union of $B_f(x)$. 
We see that a function $f$ not only defines $(d-1)$-graphs $G_f(c)=\{f=c\}$ but also $(d-2)$-graphs 
$B_f(x)$ called ``central surfaces" for every vertex $x$.  \\

If we have more than one function as constraints, the singularity structure of $\{f_1=c_1, \dots, f_k =c_k \;\}$ 
is more complicated and resembles the classical situation, where singularities can occur. Also the Lagrange setup, where
we mazimize or minimize functions under constraints, is very similar to the classical situation.
The higher complexity entering with 2 or more functions is no surprise as also
for classical algebraic sets defined as the zero locus of finitely many polynomials, the case of several functions is harder to analyze. 
In calculus, when studying extrema of a function $f$ under constraints $g$, following Lagrange, one 
is interested in critical points of $f$ and $g$ as well as places, where the gradients are parallel. 
Lets assume that we have two functions $f,g$ on the vertex sets of a geometric graph. The intersection $\{ f=c, g=d \}$ 
with $S(x)$ is then a sphere of co-dimension $2$. If we have $2$ constraints and $G$ is $4$-dimensional for example, 
then $S(x)$ is a $3$-sphere and $S(x) \cap \{ F=0 \}$ is a knot inside $S(x)$. This can already be complicated. By triangulating
Seifert surfaces associated by a knot, one can see that any knot can occur as a co-dimension $2$ curve of a graph. 
Briskorn sphere examples allow to make this explicit in the case of torus knots. 
With more than one function, Poincar\'e-Hopf indices form a discrete $1$-form valued grid because changing any of the $c_j$ can change
the Euler characteristic.  This allows to express the Euler characteristic as a discrete line integral in the discrete target set of $F$.
There are now many Poincar\'e-Hopf theorems, for every deformation path, there is one.  \\

If we look at the set $F=c$, where all functions change sign,
singular values for $F=(f_1, \dots ,f_k)$ are vectors $c=(c_1,\dots,c_k)$ in $R^k$ for which the graph $\{ F=c \}$ is 
not a $(d-k)$-graph or values $c=F(x)$ taken on by $F$ on vertices $x$. 
Lets look at the very special example, where all $f_j$ are the same function.
Now, near the diagonal $c_i=c$, there is an entire neighborhood of parameter values, where regularity fails. 
We see that the Sard statement does not hold in this commutative setting. We therefore also look at the non-commutative 
setup, where we fix $f_1=c_1$ first, then look at the level surface $\{ f_2=c_2 \}$ on the surface $\{f_1=c_1 \}$ and proceed inductively. 
Sard is now more obvious, but the sets $F=0$ depend on the order, in which the functions $f_j$ have been chosen. 
This order dependence is no surprise in a quantum setting, if we look at $f_j$ as observables. It simply depends in 
which order we measure and fix the $f_j$.  \\

There are analogies to Morse theory in the continuum: 
in the case of one single function, we can single out a nice class of graphs and functions, which 
lead to a discrete Morse theory. The geometric graph $G$ as well as functions on vertices are the 
only ingredients. One can now assign a Morse index $m(x)$ at a critical point and have $i_f(x) = (-1)^{m(x)}$. 
The requirement is an adaptation of a reformulation of the definition of being Morse means in the continuum that for a small
sphere $S(x)$, the set $S(x) \cap \{ f(y)=f(x) \}$ is a product $S_{m-1} \times S_{d-1-m}$. For example
for $m=1,d=2$, we have a saddle point, where $S(x)$ intersects the level surface in $4$ points forming 
$S_0 \times S_0$. In the enhanced picture, where we look at the graph $G_1$, we have this regularity 
more likely. One can extend the function to the new graph and repeat until one get a Morse function. \\

We can use graphs described by finitely many equations in order to construct examples of graphs
to illustrate classical calculus like Stokes or Gauss theorem or surfaces to illustrate 
classical multivariable calculus. 
The notion of $d$-graphs has evolved from \cite{elemente11,indexexpectation,eveneuler,knillgraphcoloring} 
to \cite{knillgraphcoloring2}, where it reached its final form. While finishing up \cite{KnillJordan} a literature
search showed that spheres had been defined in a similar way already by Evako earlier on. \\

The enhanced Barycentric refinement graph $G_1 = G \times K_1$ is a regularizes graph which helps to study 
Jordan-Brouwer questions in graph theory \cite{KnillJordan} and introduce a product structure
on graphs which is compatible with cohomology \cite{KnillKuenneth}. It also illustrates
the Brouwer fixed point theorem \cite{brouwergraph}, as the fixed simplices on $G$ can be seen as fixed points on $G_1$. 
The simplex picture is also useful as the Dirac operator $D=d+d^*$ on a graph builds on it
\cite{DiracKnill,knillmckeansinger}. The graph spectra of successive Barycentric refinements converges
universally, only depending on the size of the largest complete subgraph \cite{KnillBarycentric}.\\

An other application of the present Sard analysis is a simplification of
\cite{indexformula} which assures that the curvature $K$ is identically zero for 
$d=(2m+1)$-dimensional geometric graphs: write the symmetric index $j_f(x)=i_f(x)+i_{-f}(x)$ in terms of
the Euler characteristic of the $(d-2)$-dimensional graph $G_f(x)$ obtained by taking the hypersurface $\{ y \; | \;  f(y)=f(x) \}$
in the unit sphere $S(x)$. For odd-dimensional geometric $d$-graphs,
this symmetric index is is $-\chi(B_f(x))/2=j_f(x)$, while for even-dimensional graphs, it is $1-\chi(B_f(x))/2=j_f(x)$. 
One of the corollaries given here is that if $f$ is locally injective, then
$B_f(x)$ is always a geometric $(d-2)$-graph if $G$ is a $d$-graph. 
In \cite{indexexpectation}, we called the uncompleted graphs polytopes and showed that one can complete them. 
Now, since the expectation of $j_f(x)$ is curvature, \cite{indexexpectation,colorcurvature},
this gives an immediate proof that odd-dimensional $d$-graphs have constant $0$ curvature. 
Zero curvature follows for odd-dimensional graphs. The regularity allows to interpret Euler curvature as 
an average of two-dimensional sectional curvatures. Euler characteristic written as the expectation
of these sectional curvature averages is close to Hilbert action. It shows that the quantized 
functional ``Euler characteristic" is not only geometrically relevant but that it has physical potential. 

\section{Level surfaces} 

The study of level surfaces $\{ f=c \}$ in a graph is not only part of discrete differential topology.
It belongs already to discretized multivariable calculus, where surfaces $\{ f(x,y,z) = 0 \}$ in space or 
curves $\{ f(x,y) = 0 \}$ in the plane are central objects. Some would call calculus on graphs ``quantum calculus".
We want to understand under which conditions a sequence of $k$ real valued functions 
$F=(f_1,\dots,f_k)$ on the vertex set of a graph leads to co-dimension $k$ graphs 
$\{ F=0 \}$. We start with the case $k=1$, where we have a level surface $\{ f = c \}$.

\begin{defn}
Given a finite simple graph $G=(V,E)$, 
define the {\bf level hyper surface} $\{ f=c \}$ as the graph $G'=(V',E')$ for which the 
vertex set $V'$ is the set of simplices $x$ in $G$, on which the function $f-c$ changes sign in the 
sense that there are vertices in $x$ for which $f-c < 0$ and vertices in $x$ for which $f-c>0$. 
A pair of simplices $(x,y) \in V' \times V'$ is in the edge set $E'$ of $G'$ if $x$ is a subgraph 
of $y$ or if $y$ is a subgraph of $x$. 
In the case $c=0$, the graph $\{ f = 0 \}$ is also called the {\bf zero locus} of $f$.
\end{defn}

{\bf Remarks}. \\
{\bf 1)} Instead of taking a real-valued function, we could take a function taking values
in an ordered field. We need an ordering as we need to tell, where a function ``changes sign".  \\
{\bf 2)} Most of the time we will assume $c$ is not a value taken by $f$. An alternative 
would be to include all simplices which contain a vertex on which $f=0$. 

\begin{figure}[ph]
\scalebox{0.2}{\includegraphics{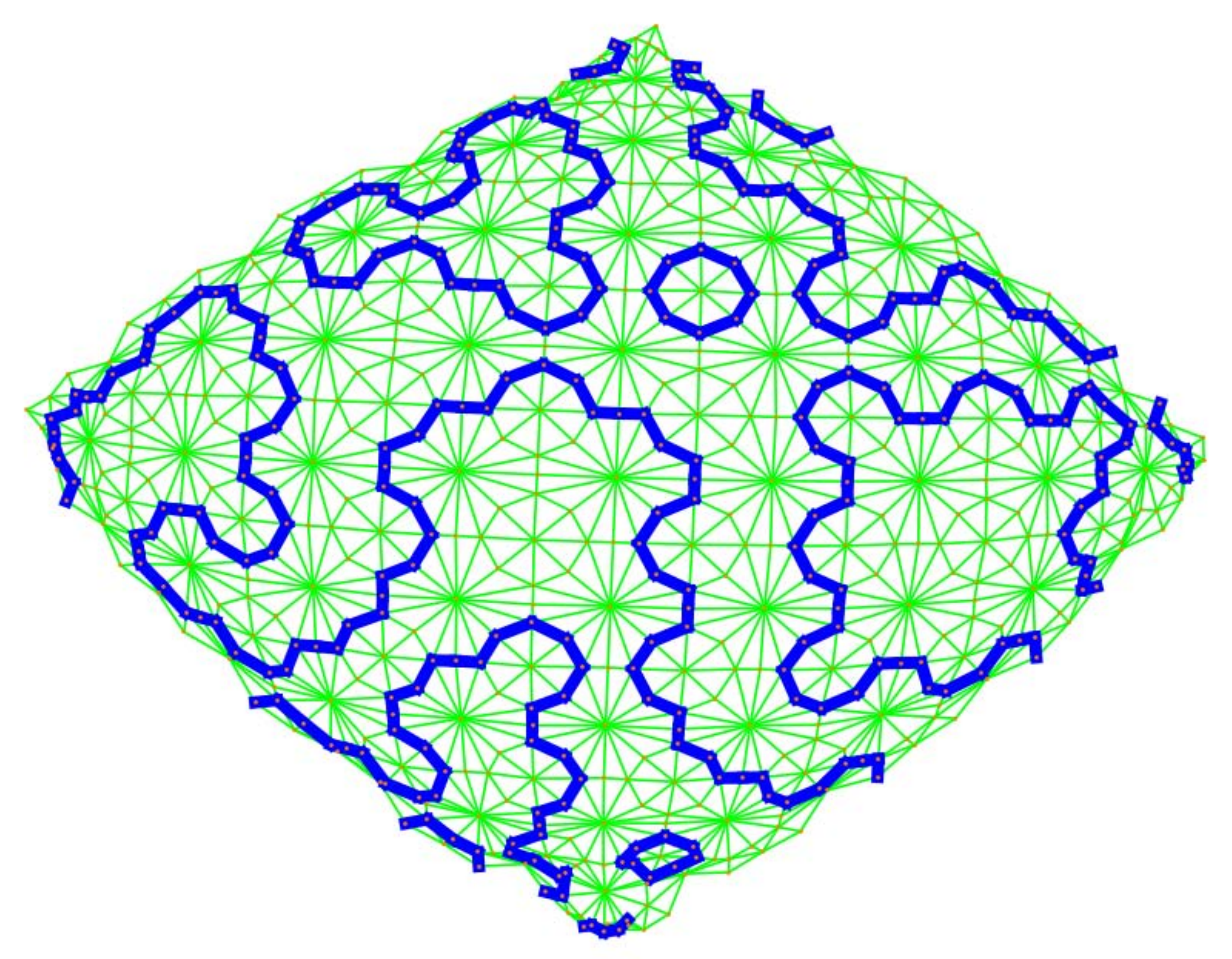}}
\scalebox{0.2}{\includegraphics{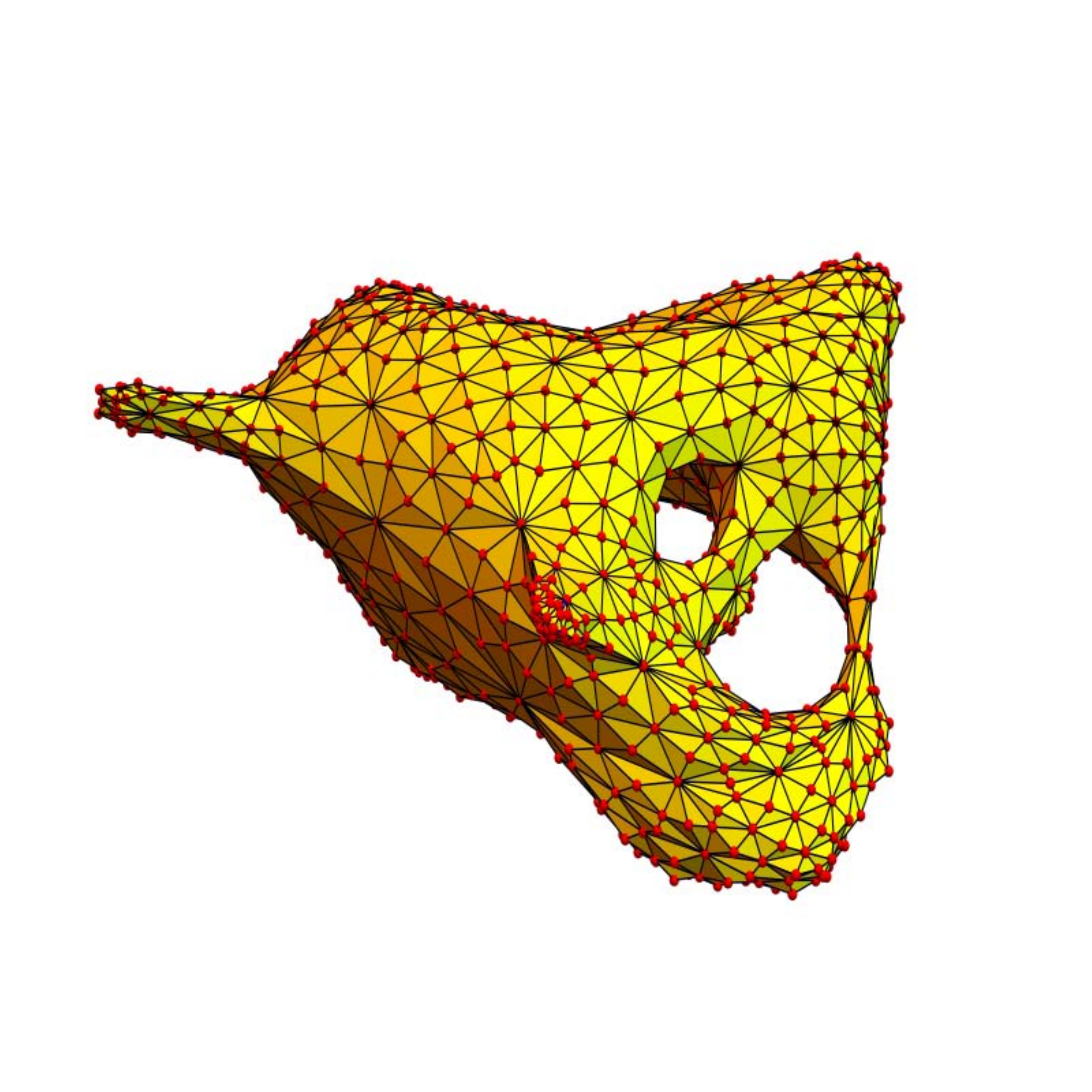}}
\caption{
A level curve in a planar graph. A level surface in a 3-sphere.
}
\end{figure}

{\bf Examples.} \\
{\bf 1)} Let $G$ be a $2$-sphere like for example an icosahedron. Let $f$ be $1$ on exactly one vertex $x$
(a discrete Dirac delta function) and $0$ everywhere else. Now $\{ f=1/2 \}$ consists of all
edges and triangles containing $x$. They form a circular graph and $\{ f = 0 \} = C_{2 {\rm deg}(x)}$. \\
{\bf 2)} Let $G= C_n \times C_n \times C_n$ be a discrete $3$-torus and let $f$ be a function which is $1$
on a circular closed graph and $0$ else. Then $ \{ f =1/2 \}$ is a $2$-dimensional torus. 

\section{The Sard lemma} 

The following definitions are recursive and were first put forward by Evako. 
See \cite{KnillJordan} for our final version.  

\begin{defn}
A {\bf $d$-sphere} is a finite simple graph for which unit sphere $S(x)$ 
is a $(d-1)$-sphere and such that removing a single vertex from the graph 
renders the graph contractible. Inductively, a graph is {\bf contractible}, if there exists
a vertex such that both $S(x)$ and the graph generated with vertices without $x$
are contractible. 
\end{defn}

\begin{defn}
A {\bf $d$-graph} $G$ is a finite simple graph for which every unit
sphere $S(x)$ is a $(d-1)$-sphere. 
\end{defn}

{\bf Examples.} \\
{\bf 1)} A $1$-graph is a finite union of circular graphs, for which each connectivity 
component has $4$ or more vertices.  \\
{\bf 2)} The icosahedron and octahedron graphs are both $2$-graphs. In the first
case, the unit spheres are $C_5$, in the second case, the unit spheres are $C_4$. \\
{\bf 3)} In \cite{KnillEulerian} we have classified all {\bf Platonic d-graphs} using
Gauss Bonnet \cite{cherngaussbonnet}. Inductively, a $d$-graph $G$ is called Platonic, if there exists a 
$(d-1)$-graph $H$ which is Platonic such that all unit spheres of $G$ are isomorphic to $H$. 
In dimension $d=3$, there are only two Platonic graphs, the 16 and 600 cell. 
In dimensions $d \geq 4$, only the cross polytopes are platonic.  \\

The following Sard lemma shows that we do not have to check 
for the geometric condition if we look at level surfaces: it is guaranteed, as
long as we avoid function values in $f(V)$. Not even local injectivity is needed: 

\begin{lemma}[Sard lemma for real valued functions]
\label{sardlemma}
Given a function $f:V \to {\bf R}$ on the vertex set of a $d$-graph $G=(V,E)$. 
For every $c \notin f(V)$, the level surface $\{ f=c \}$ is either the empty graph 
or a $(d-1)$-graph.
\end{lemma}
\begin{proof}
We have to distinguish various cases, depending on the dimension of $x$.
If $x$ is an edge, where $f$ changes sign, then $f$ changes sign on each 
simplex containing $x$. The set of these simplices is 
a unit sphere in the Barycentric refinement $G_1$ of $G$ and therefore a $(d-1)$ sphere.
If $x$ is a triangle, then there are exactly two edges contained in $x$,
on which $f$ changes sign. The sphere $S(x)$ in $\{ f=c \}$ is a suspension 
of a $(d-2)$-sphere: this is the join of $S_0$ with $S_{d-2}$. 
In general, if $x$ is a complete subgraph $K_k$ then the unit sphere $S(x)$ 
is a join of a $(k-2)$-dimensional sphere and a $(d-k)$-dimensional sphere,
which is a $(d-2)$-dimensional sphere. As each unit sphere
$S(x)$ in $\{ f = c \}$ is a $(d-2)$-sphere, the level surface 
$\{ f=c \}$ is a $(d-1)$-graph. 
\end{proof}

{\bf Examples.}  \\
{\bf 1)} If $d=2$, and $f$ changes sign on a triangle $K_3$, then it changes
sign on exactly two of its edges. If $f$ changes sign on an edge, then it 
changes sign on exactly two of its adjacent triangles. We see that the level surface
$\{ f=c \}$ is a graph for which every vertex has exactly two neighbors. In other
words, each unit sphere is the $0$-sphere. \\
{\bf 2)} If $d=3$, and $f$ changes sign on a tetrahedron $x=K_4$, then there are
two possibilities. Either $f$ changes sign on three edges connected to a vertex
in which case we have $3$ edges and $3$ triangles in the unit sphere of $K_4$
with $4$ vertices.  A second possibility is that $f$ changes on $4$ edges and $2$ triangles
in which case the unit sphere consists of $6$ vertices. 
Now look at a triangle $x=K_3$. It is contained in exactly two tetrahedra and contains
two edges. The unit sphere is $C_4$. Finally, if $x=K_2$ is an edge, then all triangles
and tetrahedra attached to $x$ form a cyclic graph of degree $C_{2n}$ where $n$ is the 
number of tetrahedra hinging on $x$.  \\

The Sard lemma can be used for minimal colorings for which the number of colorings
is exactly known:

\begin{coro}
If $G$ is a $d$-graph and $c$ is not in the range of $f$, then 
the surface $H=\{f=c\}$ is a $(d-1)$-graph which is $d$-colorable. 
The chromatic polynomial $p(x)$ of these graphs satisfies $p(d)=d!$.
\end{coro}
\begin{proof}
To every vertex of $G_1$, we can attach a ``dimension" which is the dimension
of the simplex in $G$ it came from. This dimension is the coloring. It remains
a coloring when looking at subgraphs. 
\end{proof}

{\bf Examples.} \\
{\bf 1)} For a level surface $f=0$ on a $d$-dimensional graph, we get a $(d-1)$ graph which is $(d+1)$-colorable. 
For example, for $d=2$, the graph can be colored with $3$ colors. This is minimal as any triangle
already needs $3$ colors. It implies the graph is Eulerian: the vertex degree is even everywhere. \\
{\bf 2)} If $\Omega$ is the number of colorings with minimal color of $G$, we can 
for every vertex $x$ look at the {\bf index} $i_f(x) = 1-\chi(S^-_f(x))$ where $S^-_f(x)$ is the
set of vertices $y$ on the sphere $S(x)$, where $f(y)<f(x)$.  Given a geometric graph $G$ of 
dimension $d$, then $G_1$ is $d+1$ colorable. The number of colorings is
$(d+1)!$. We can look at the list of indices which are possible on each point and call this the
index spectrum. The set of vertices $x$ where $\{ f<c \}$ changes the homotopy type are called {\bf critical points}
of $f$.  If the index is nonzero, then we have a critical point because the Euler characteristic is a homotopy invariant. 
But there are also critical points with zero index, as in the continuum. 
Finding extrema of $f$ can be done by comparing the function values of all 
vertices where $i_f(x)$ is not zero or more generally, where $S^-_f(x)$ is not
contractible. At a local minimum $S^-f(x)$ is empty.  

\section{The central surface}

Besides the surface $G_f(c) = \{ f=c \}$, there is for every vertex $x$ a central surface of co-dimension $2$ 
which is obtained by looking at level surfaces $B_f(x)$ obtained by looking at $\{ f = c \}$ inside
the unit sphere $S(x)$. This object was introduced in \cite{indexformula,eveneuler} for $4$-graphs $G$, 
where the central surface is a $2$-dimensional graph, a disjoint union of $2$-dimensional subgraphs 
$B_f(x)$ of the $3$-dimensional unit spheres $S(x)$. 

\begin{defn}
A real-valued function $f$ on the vertex set of a graph $G=(V,e)$ is called {\bf locally injective},
if $f(x) \neq f(y)$ for $(x,y) \in E$. An other word for a locally injective function is a
{\bf coloring}. 
\end{defn}

\defn{Given a $d$-graph $G$, a locally injective function $f$ and a vertex $x_0$, define 
the {\bf central surface} $B_f(x_0)$ as the 
level surface $\{ f = f(x_0) \}$ in $S(x_0)$. It is a $(d-2)$ graph. The graph 
consists of all simplices in $S(x_0)$ for which $f$ takes values smaller or 
larger than $f(x_0)$.\\ }

Each of these surfaces are subgraphs of their unit sphere $S(x_0)$. We have one surface for
each vertex $x_0$. On each sphere $S(x,y) = S(x) \cap S(y)$, we can look at the intersection of 
$B_f(x)$ and $B_f(y)$. It consists of all simplices in $S(x,y)$ where both $f(z)-f(y)$
and $f(z)-f(x)$ change sign. Sometimes they can be joined together along a circle. 
For example, given a $3$-graph $G$, then the union $B_f$ of all $B_f(x)$ consists
of all edges and triangles so that the max and min 
on each larger tetrahedron are attained in the edge or triangle. 
The following definition was first done in \cite{poincarehopf}:

\begin{defn}
Given a finite simple graph $G$ and a locally injective real-valued function $f$ on the vertex
set $V$, the {\bf Poincar\'e-Hopf index} is defined as
$i_f(x) = 1-\chi(S^-_f(x))$, where $S^-_f(x)$ is generated by $\{ y \in S(x) \; | \; f(y)<f(x) \}$. 
The {\bf symmetric index} is defined as $j_f(x) = (i_f(x) + i_{-f}(x))/2$.
\end{defn}

The Poincar\'e-Hopf theorem \cite{poincarehopf} 
tells that $\sum_{x \in V} i_f(x)=\chi(G)$. Since this also holds for the function
$-f$, we have 
$$ \sum_{x \in V} j_f(x) = \chi(G) \; . $$
The following remark made in \cite{indexformula} expresses $j_f(x)$ as the 
Euler characteristic of a central surface, provided the graph is geometric: 

\begin{propo}
Given a $d$-graph $G$ and a locally injective function $f$. If $B_f(x)$ 
is the central surface, then  for odd $d$, we have 
$$ j_f(x) = -\chi(B_f(x))/2 \; , $$
for even $d$, we have 
$$ j_f(x) = 1-\chi(B_f(x))/2 \; . $$
\end{propo}

\begin{proof}
$B_f(x)$ is a $(d-2)$-graph in $S(x)$, which by assumption is a $(d-1)$-graph. 
Since $f$ is locally injective, the function $g(y)=f(y)-f(x)$ does not take
the value $0$ on $S(x)$. By the Sard lemma, the graph $B_f(x)$ is a $(d-2)$-graph.
\end{proof}

\begin{coro}
For a $d$-graph $G$ with odd dimension $d$, the curvature 
$$  K(x) = 1- \frac{S_0(x)}{2} + \frac{S_1(x)}{3}- \frac{S_2(x)}{4} + \dots $$ 
(where $S_k(x)$ are the number of $K_{k+1}$ subgraphs of the unit sphere $S(x)$)
has the property that it is constant zero for every vertex $x$. 
\end{coro}

\begin{proof}
The expectation $j_f(x)$ is curvature \cite{indexexpectation,colorcurvature}.
Since $j_f(x)$ is identically zero as the Euler characteristic of an odd dimensional 
$(d-2)$-graph, also curvature is identically zero. 
\end{proof} 

Note that in the continuum, the Euler curvature is not even defined for odd dimensional 
graphs as the definition involves a Pfaffian \cite{Cycon}. Having the value $0$ in the 
discrete is only natural. \\

\begin{defn}
A function $f$ on the vertex set is called a {\bf Morse function} if it is locally injective and if at 
every critical point, there is a positive integer $m$, such that $B_f(x) = \{ f(y)=f(x) \}$ within $S(x)$ is a 
product $S_{m-1} \times S_{d-1-m}$ or the empty graph if $m=0$ or $m=d$. 
The integer $m$ is called {\bf Morse index} of the critical point $x$.
\end{defn}

Depending on whether $m$ is odd or even, we have $\chi(B_f(x)) = 4$ or $0$ so
that the index $1-\chi(B_f(x))/2=-1$ if $m$ is odd and $1$ if $m$ is even.
When adding a critical point, this corresponds to add a $m$-dimensional handle.
It changes the Euler characteristic by $(-1)^m$ and changes the $m$'th cohomology by $1$.

\section{Lagrange} 

In this section we try to follow some of the standard calculus setup when extremizing
functions with or without constraints. But it is done in a discrete setting, where space is 
a graph. As school calculus mostly deals 
with functions of two variables, we illustrate things primarily for $2$-dimensional graphs, 
even so everything can be done in any dimensions.  \\

There are three topics related to critical points in two dimensions: A) extremizations without constraints, 
B) equilibrium points of vector fields and C) extremization problems with constraints which are called Lagrange problems.
In the case A), can look at extrema of a function $f$ on the vertex set of a $2$-graph, 
in the case B) we look at equilibrium points of a pair $F=(f,g)$ of functions on the vertex set of a $2$-graph,
and finally in the case C), we look at extrema of a function $f$ on the vertex set under the constraint 
$g=c$ on a $2$-graph. \\

Lets first look at the ``second derivative test" on graphs. Recall that
a vertex $x$ in a graph is a critical point of a function $f$, if $\{ f<f(x) \}$ and $\{f \leq f(x) \}$ are not
homotopic. This is equivalent to the statement that 
$S^-_f(x) = \{ y  \in S(x)\; | \; f(y)<f(x) \}$ is a graph which is not contractible. 
The analogue of the discriminant $D$ is the Poincar\'e-Hopf index $i_f(x) = 1-\chi(S^-_f(x))$. 
Here is the analogue of the second derivative test:

\begin{propo}[Second derivative test]
Let $G$ be a $d$-graph and assume $f$ is locally injective and $x$ is a critical point. There
are three possibilities: \\
a) If $S^-_f(x)$ is a $(d-1)$-sphere, then $x$ is a local maximum.  \\
b) If $i_f(x)$ is positive and $S^-_f(x)$ is empty then $x$ is a local minimum.  \\
c) If $d=2$ and if $i_f(x)$ is negative, then $x$ is a type of saddle point. 
\end{propo}

\begin{proof}
For a $2$-dimensional graph, the index is nonzero if and only if $x$ is a critical point because
a subgraph of a circular graph has Euler characteristic $1$
if and only if it is a contractible graph. It has Euler characteristic $0$ if and only if it is either
empty or the full circular graph. In all other cases of subgraphs of a circular graph, 
the Euler characteristic counts the number of connectivity components.  \\
In higher dimensions, there are cases of graphs $S_f^-(x)$ having Euler characteristic $1$ but
not being contractible. In higher dimensions, $S_f^-(x)$ can have negative Euler characteristic
so that the index $i_f(x)$ can become larger than $1$. 
\end{proof}

{\bf Example.} \\
{\bf 1)} For $d=2$, the standard saddle point is $i_f(x)=-1$. The function $f$ changes sign on 4 points.
A discrete ``Monkey saddle" has index $i_f(x)=-2$. 
It is obtained for example at a vertex $x$ for which the unit ball is a wheel graph with $C_6$ boundary such that
$f(y)$ is alternating smaller or bigger than $f(x)$. \\
{\bf 2)} If $d$ is odd and $f$ has a local maximum, $i_f(x)=-1$. This is analogue to the continuum, where 
$D$ is the determinant of the Hessian. \\

For simplicity, we restrict to a simple $2$-dimensional situation, where
we have two functions $f,g$ on the vertex set $V$ of a $2$-graph $G=(V,E)$. 
We can think of $F=(f,g)$ it as a vector field and see the equilibrium points are the intersection
of null-clines as in the continuum. Classically, the critical 
points of $f$ under the constraint $g=c$ 
are the places, where these null-clines are tangent or are degenerate in that one of the
gradients is zero. 

\defn{
Let $F=(f,g)$ be two functions on the vertex set of a $2$-graph. 
The set of equilibrium points $\{ (x,y) \; | \; F(x,y)=(0,0) \; \}$ is a
zero-dimensional graph given by the set of triangles, where $f$ 
and $g$ both change sign.}

In more generality we have defined the set $G_F=\{ f_1=0, \dots ,f_k=0 \}$ as the graph whose vertex set
consists of the set of simplices of dimension in $\{ k,\dots ,d\}$ on which all functions $f_j$ change sign
and where two simplices are connected, if one is contained in the other. It is possible for 
example that for a $d$-graph, the set $G_F$ consists of {\bf all} 
$d$-dimensional simplices in $G$. This is still a $0$-graph, a graph with no edges because 
no $d$-simplex is contained in any other $d$-simplex. 

\defn{
Let $G$ be a $2$-graph and let $f$ be a locally injective function on the vertex set of $G$.
Define the {\bf continuous gradient} of $f$ in a triangle $t=(xyz)$ with origin $x$ as 
$df=\langle f(y)-f(x),f(z)-f(x)) \rangle$. This 
assigns to each triangle in $G$ a vector with two real components. 
By taking signs, the vector $\nabla f$ becomes an element in the finite vector space
$Z_2^2$. This is the {\bf discrete gradient} on the triangle $t$ rooted 
at the vertex $x$. \\
}

{\bf Example.} \\
{\bf 1)} If $G$ is a $2$-graph and $x$ belongs to a triangle $t$ containing two edges
which both contain $x$, then $\nabla f(x,t) = \langle 1,1 \rangle$. Two functions have
a parallel gradient in $t$, if and only the sign changes in $t$ happen
on the same edges of the triangle. \\

{\bf Remarks}: \\
{\bf 1)} A function $f$ on the vertex set $V$ is a $0$-form. The exterior derivative
$df$ is a function is a 1-form, a function  on the edges of the graph. The exterior derivative 
depends on a choice of orientations on complete subgraphs. For $1$ forms in particular, where edges
have been ordered at first,
the situation at a vertex $x$ defines then $df( (x,y)) = f(y)-f(x)$. When restricting to a $d$-simplex, we get
$d$ real numbers, we as in the continuous forms the {\bf gradient}. \\
{\bf 2)} The ordering of the coefficients of the gradient vector 
depends on the orientation of the triangle. 
This is similar to the classical case, where a triangle in a triangulation of a surface
defines a normal vector at $x$, once a vertex $x$ of the triangle and an orientation of 
the triangle is given. \\
{\bf 3)} The analogue of rank $d(f,g)=2$ means that simultaneous sign changes happen on one edge of
a triangle only, so that the two sign change is in the same direction on that edge. 
Geometrically this implies that the spheres $f=0,g=0$ in $S(x)$ intersect transversely. 
Two discrete gradients in $Z_2^d$ are parallel if and only they are the same because
the only nonzero scalar is $1$. \\

The classical Lagrange analysis shows that the critical points of $F$, the place where the Jacobean 
$dF$ has rank $0$ or $1$ are candidates for maxima of $f$ under the constraint $g=c$.
In the same way as in the continuum, we can write down Lagrange equations for any number 
of functions on a $d$-graph. The most familiar case for two functions: 

\begin{propo}
Let $G$ be a $d$-graph. 
If the discrete gradients $\nabla f$, $\nabla g$ are nowhere parallel, then 
$F = \{ f = 0, g = 0 \}$ is a $(d-2)$-graph. 
\end{propo}

It should be possible to estimate the measure of the set of global critical values
if $F=(f_1,\dots,f_k)$ are functions on the vertex set of a finite simple graph $G$. 
Instead of doing so, we will look later at a Sard setup for which the critical 
values have zero measure. \\

Lets look at the example of two functions $f,g$ on a $3$-graph.
The set $\{ F=0 \}$ is the set of simplices, where both $f,g$ change sign.
The condition $dF$ having maximal rank means that the surface $f=0$ on $S(x)$
and the sphere $g=0$ in $S(x)$ intersect in a union of $1$-spheres.
Here is the Lagrange setup with two functions in three dimensions: 

\begin{propo} If $F=(f,g)$ are two functions on the vertex set of a $3$-graph
and if $dF$ has maximal rank at every $x$ and $(c,d)$ is not in $F(V)$,
then $F=c$ is either empty of a $1$-graph, a finite collection of circles.
\label{case2-3}
\end{propo}

\begin{proof}
The set $\{ F=c \}$ is a graph whose vertices are the triangles ($K_3$ subgraphs of $G$)
or tetrahedra  ($K_4$ subgraphs of $G$), where both $f$ and $g$ change sign. 
Let $x$ be a tetrahedron in $\{F =c \}$. 
The maximal rank condition prevents parallel gradients like 
$\nabla f=\langle 1,1,1 \rangle = \nabla g = \langle 1,1,1 \rangle$
so that it is impossible to have 3 triangles in $\{ F=c \}$ inside $x$.
Assume a single triangle is present where both $f,g$ change sign, then there
is a common edge, where both $f,g$ change sign and a second triangle must
also be in $\{ F = c \}$. We see that there are exactly two triangles $y,z$ present
on which both $f,g$ change sign. This means the vertex $x$ has exactly two neighbors
$y,z$. Each of the triangles $y,z$ has two neighbors, the tetrahedra attached to them.
We see that $\{ F =c \}$ is a graph with the property that every unit sphere is the
zero sphere $S_0$. Therefore, it is a finite collection of circular graphs. 
\end{proof}

The Lagrange problem extremizing $f$ under the constraint $g=c$ is the  following:

\begin{propo}[Lagrange]
Given two functions $f,g$ on a $2$-graph. 
Extremizing $f$ under the constraint $g=c$ happens on triangles, where
$df$ and $dg$ are parallel (which includes the case when $df=0$ or $dg=0$). 
In other words, extrema happen on Lagrange critical points. 
\end{propo}

The following maximal rank condition is the same as in the continuum. It tells
that the graph formed by the simplices having dimension in $\{ k, \dots,  d \}$ 
on which all functions $f_k$ change sign simultaneously is a $(n-k)$-graph
if the rank of the set of gradients in $Z_2^k$ is maximal: 

\begin{thm}[Regularity in the commutative setup]
Assume $G$ is a $d$-graph and $F$ is a $R^k$-valued function on the vertex set. 
If $\nabla F$ has maximal rank on every $d$-simplex, then $\{ F=c \}$ is a 
$(d-k)$-graph. 
\end{thm}
\begin{proof}
For a fixed vertex $x$ and $d$-simplex $X$, we have a tangent space $Z_2^d$
for which every function $f_j$ contributes a vector $\nabla f_i$ telling on
which edges emanating from $x$ inside $X$, the sign of $f_i$ changes. By 
assumption, the $k$ vectors $\nabla f_i$ are linearly independent.  \\
Use induction with respect to $d$: if $k=d$, there is nothing to show because
by definition, the set $\{ F=c \}$ is the set of $d$-simplices, where all functions
change sign and this is a $0$-graph as there are no edges.
We claim that in the unit sphere $S(x)$, the functions $f_1,\dots ,f_k$ induce
$k$ linearly independent gradients $\nabla f_j$ on the $(d-1)$-simplex $Y=X \setminus x$. 
Indeed, on each triangle, the sum of the $df_j$ values is zero (as ${\rm curl}({\rm grad}(f))=0$).
A nontrivial relation between the gradients on $Y$ would induced a nontrivial 
additive relation between the gradients on $X$. 
Now, by induction, the $k$ functions on $S(x)$ define a $(d-1-k)$-graph. 
\end{proof}

{\bf Remark.} One could also try induction with respect to $k$ and consider
$\{ f_1 = c_1 \}$ which is a $(d-1)$-graph, by the Sard lemma. 
The problem is that one has to extend $f_2,\dots,f_d$ in such a way on the simplices
so that one has still maximal rank condition. The problem is that $\{ f_1 = c_1 \}$
has now a different vertex set than $G$ and that linking things is difficult. \\

{\bf Examples:} \\

{\bf 1)} The case $d=3$, $k=2$ was discussed in Proposition~\ref{case2-3}. \\
{\bf 2)} In the case $d=4$, $k=2$, we want the two functions both to be locally
injective and the gradients of the two functions $f,g$ not to be parallel. 
The set $F=c$ consists of all tetrahedra $K_4$ and hypertetrahedra $K_5$ on which 
both $f$ and $g$ change sign. The gradients restricted to the tangent space on $S(x)$
are not parallel and we can apply the analysis of the previous case to each unit sphere
which shows that in $S(x)$ the set $F=c$ is a collection of circular graphs. 
This shows that the unit spheres of $F=c$ have the property that each unit sphere there
is a circular graph. Therefore $F=c$ is a geometric 2-graph, a surface if $\nabla f,\nabla g$
are nowhere parallel. \\
{\bf 3)} In the case $d=4$, $k=3$, it is the first time that the maximal rank condition is
not just a parallel condition. Given a vertex $x$ and a tetrahedron $t=(x,y_1,y_2,y_3)$. 
The discrete gradient is $\langle f(y_1)-f(x), f(y_2)-f(x), f(y_3)-f(x) \rangle$. An example
of a violation of the maximal rank condition for three functions $f,g,h$ would be 
$\nabla f = \langle 1,0,1 \rangle$, $\nabla g = \langle 1,1,0 \rangle$, $\nabla h = \langle 0,1,1 \rangle$.
They are pairwise not parallel but $\nabla f + \nabla g + \nabla h = 0$. \\
{\bf 4)} The non-degeneracy condition is not always needed: 
for $d$ functions $F=(f_1,\dots,f_d)$ on a $d$-graph, we look at the simplices
on which all functions change sign. The condition $dF \neq 0$ implies
that the set of $d$-dimensional simplices on which all functions $f_j$ 
change sign are isolated. 

\section{Sard theorem} 

Since the set of critical values can hove positive measure if we look at the simultaneous
solution, we change the setup and look at hypersurfaces in hypersurfaces. This will lead 
to a Sard theorem as in the continuum. 
We will have to pay a prize: the order with which we chose the hypersurfaces within hypersurfaces
now will matter. 

\begin{defn}
A function $f$ on the vertex set is called {\bf strongly injective} if
all function values $f(x_i)$ are rationally independent. 
It is {\bf strongly locally injective} if in each complete subgraph, the values 
are rationally independent. A list of functions $f_1,\dots, f_k$ is called 
{\bf strongly injective} if the union of all function values $f_j(x_i)$ 
are rationally independent. 
\end{defn}

Strongly injective functions are generic from the measure and Baire point of view: 
given a finite simple graph $G$ with $n$ vertices.
Look at the probability space $\Omega$ of all functions from the vertex set to $[-1,1]$, 
where the probability measure is the product measure on $[-1,1]^n$.

\begin{lemma}
For any $k$, with probability one, a random sample $f_1,\dots,f_k$ in $\Omega^k$ is strongly 
locally injective.
\end{lemma}
\begin{proof}
There are $v=v_0+v_1+\cdots+v_d$ complete subgraphs in $G$. They define the $\sum_i i \cdot v_i$
numbers $c_{ij} = f_i(x_j)$. There is a countably many rational independence conditions $\sum_{ij} a_{ij} c_{ij} =0$
to be avoided, where $a_{ij}$ are integers. The complement of a countable union of such hyperplane sets of zero measure in $[-1,1]^n$
and consequently has zero measure.
\end{proof}

\begin{defn}
Given an ordered list of functions $f_1,f_2,\dots, f_k$ from the vertex list $V$ of a finite simple
graph $G=(V,E)$ and $c_1, \dots, c_k$ be $k$ values. Let $\overline{f}_1=f_1$. Denote by $\overline{f}_2$ the function 
$f_2$ extended to $\{ \overline{f}_1 = c_1 \}$ and by $\overline{f}_3$ the function $f_3$ extended to 
$\{ \overline{f}_1 = c_1, \overline{f}_2 = c_2 \}$ etc, always assuming that $\overline{f}_{j+1}$ defined on
$\overline{f}_1 = c_1, \dots, \overline{f}_j = c_j \}$ does not take the value $c_{j+1}$. We call
the sequence $c_1,\dots, c_k$ {\bf compatible} with $f_1,\dots,f_k$ if a sequence $\overline{f}_j$ can be 
defined so that none of them are not constant. 
\end{defn}

\begin{thm}[Discrete Sard for an ordered set]
Given $k$ strongly injective functions $f_{1},\dots,f_k$. For all except a finite set of 
vectors $(c_1,\dots,c_k)$ the sequence $c_1,\dots,c_k$  is compatible and 
the set $\{ f_1 = c_1, \dots, f_k = c_k \}$ is a geometric $(d-k)$-graph. 
\end{thm}

\begin{proof}
This follows inductively from the construction. In each step, only a finite set of
$c$ values are excluded. 
\end{proof}

The assumption is stronger than what we need. It sometimes even works in 
the extreme case of two identical functions: Let $f$ be the function
on the octahedron given the values $f_1(1)=13,f_1(2)=15,f_1(3)=17,f_1(4)=19$ on the equator and 
the value $f_1(5)=1$ on the north pole and the value $f_1(6)=31$ on the south pole. 
Lets take $c_1=2$.  Now, $\{ f_1 = c_1 \}$ is the cyclic graph with vertices 
$\{ (51)$, $(512)$, $(52)$, $(523)$, $(53)$, $(534)$, $(54)$, $(541) \}$. The function
$\overline{f}_2$ takes there the values 
$\overline{f}_2(51)=(f_2(5)+f_2(1))/2=(1+13)/2=7$,
$\overline{f}_2(52)=(f_2(5)+f_2(2))/2=(1+15)/2=8$,
$\overline{f}_2(53)=(f_2(5)+f_2(3))/2=(1+17)/2=9$,
$\overline{f}_2(54)=(f_2(5)+f_2(4))/2=(1+19)/2=10$,
$\overline{f}_2(512)=(f_2(5)+f_2(1) f_2(2))/2=(1+13+15)/3=29/3$,
$\overline{f}_2(523)=(f_2(5)+f_2(2) f_2(3))/2=(1+15+17)/3=35/3$,
$\overline{f}_2(534)=(f_2(5)+f_2(3) f_2(4))/2=(1+17+19)/3=37/3$,
$\overline{f}_2(523)=(f_2(5)+f_2(2) f_2(3))/2=(1+19+13)/3=11$. Now for example, for $c_2=8.5$
the set $\{ f_2=c_2 \}$ is a $0$-graph.  \\

As an example, lets look at the {\bf double nodal surface} $f_3 =0$ in $f_2=0$,
where $f_3$ is the third eigenvector. By Sard, we know:

\begin{coro}
If $G$ is a $d$-graph and the eigenfunctions $f_2,f_3$ are strongly injective not
having the value $0$, then the double nodal surface is a $(d-2)$-surface in $G_2$.
\end{coro}

{\bf Example}: \\
1) For the octahedron, the smallest $2$-sphere, the spectrum is of the Laplacian is
$$  \{ 0,4,4,4,6,6 \}  $$
The ground state space is the eigenspace to the eigenvalue $4$ and three dimensional
and spanned by $(-1,0,0,0,0,1)$, $(0,-1,0,0,1,0)$, $(0,0,-1,1,0,0) \}$.
The eigenspace to the eigenvalue $6$ is spanned by $(1,0,-1,-1,0,1), (0,1,-1,-1,1,0)$.
In both eigenspaces, there are injective functions
$f_2=(-1,-2,-3,3,2,1),f_3=(1,2,-3,-3,2,1)$ in the eigenspace. 
The graph $\{ f_2=0\}$ is the cyclic graph $C_{12}$ while
the graph $\{ f_3 =0 \}$ is $C_8 \cup C_8$. The graph $\{ f_3=0\}$ within $\{f_2=0\}$
is a two point graph. The graph $\{ f_2=0 \}$ within $\{ f_3=0 \}$ is not defined
as $f_2$ extended to the simplex set in a linear way produces a lot of function values
$0$. 


It leads to a generalization of a result we have shown for $2$-graphs:

\begin{coro}
Any compact $d$-manifold $M$ has a finite triangulation which is $(d+1)$-colorable.
\end{coro}

\begin{proof}
By Nash-Tognoli, any compact manifold can be written as a variety $F=c$ in some $R^d$.
Now just rewrite this in the discrete as the zero locus of $F=c$.
\end{proof}

\begin{coro}
The curvature at a vertex $x$ of such a graph triangulation can be written as the expectation of
$d!$ Poincar\'e-Hopf indices.
\end{coro}

If $\{f=0\}$ be the zero locus of $f:G \to R$. It would be nice to
find a smaller homeomorphic graph which represents this set.

\section{Nodal sets}

To illustrate a possible application, 
lets look at the problem of nodal sets of the Laplacian $L$ of a graph. Understanding
the {\bf Chladni patterns} of the Laplacian on a manifold with or without boundary 
is a classical problem in analysis. The nodal region theorem of Courant in the discrete
also follows from the min-max principle
\cite{Fiedler1975,VerdiereGraphSpectra,Spielman2009}. For a general graph $G$, and an eigenfunction $f$,
one looks at the number of connectivity components of $Z_f^+ = \{ \pm f \geq 0 \}$.
Let $v_k$ be the $k$'th eigenvector. Since both for compact Riemannian manifolds as well as for finite graphs, 
the zero eigenvalues are not interesting as harmonic functions are locally constant, we look primarily at the
second or third eigenvalue.  The Fiedler nodal theorem assures then that the graph generated by $\vec{v}_k > 0$ 
has maximally $k-1$ components. Especially, the second eigenvector, the ``ground state", always has 
exactly two nodal components. If an eigenfunction $f$ of the Laplacian is locally injective and has no roots, 
we can look at its nodal surfaces $f=0$ separating the nodal regions. 
I $f$ does not take the value $0$, then $\{f = 0 \}$ is defined even if $f$ is not locally injective.
As in the continuum, one can ask how big the set $\{ v \in V \; | \; f(v)=0 \}$ can become. 

If $G$ is a $2$-sphere, then the two nodal surface is a simple closed curve in $G$. 
How common is the situation that the ground state does not take the value $0$ and is 
locally injective? The situation that $0$ is in the range of the ground state appears
to be rare. For random 2-spheres (of the order of 500 vertices 
generated by random edge refinements from platonic and Archimedean solids)
we get a typical ground state energy in the order of $0.08$ and the third eigenvalue in
the order $0.2$. \\

\cite{BLS} note the following {\bf eigenfunction principle}: 
any eigenfunction to an eigenvalue $0<\lambda<n$
takes the value $0$ on every vertex of degree $n-1$, if $n$ is the number of vertices:
Proof: let $f$ be the eigenfunction to an eigenvalue $\lambda$. 
The function $f$ is perpendicular to the harmonic constant function so that 
$\sum_{x \in V} f(x)=0$. From $Lf(v) = (n-1) f(v) - \sum_{x \neq v} f(x) = \lambda f(v)$, we get
$n f(v) = \lambda f(v)$ which by assumption implies $f(v)= 0$.  \\

{\bf Example.} \\
For a wheel graph with $n$ vertices, there is one eigenvalue $n$ with
eigenvector $(1-n,1,1, \dots, 1)$, an eigenvalue $0$ with eigenvector $(1,1, \dots, 1)$.
All eigenfunctions to eigenvalues between take the zero value somewhere. For example, 
in the case $n=7$, the eigenvalues are $\{ 0,2,2,4,4,5,7 \}$ with 
eigenvectors $[0,-1,-1,0,1,1,0]$, $[0,1,0,-1,-1,0,1]$, 
$[0,-1,1,0,-1,1,0]$, $[0, -1, 0, 1,-1,0,1]$,
$[0,-1,1,-1,1,-1,1]$ and $[-6,1,1,1,1,1,1]$.  \\

Let $G$ be a $d$-graph.
Let $f$ be the ground state, the eigenvector to the first nonzero eigenvalue $\lambda$, the spectral gap. Let
$d$ denote the exterior derivative. There will be no confusion with $d$ also denoting the dimension of $G$. 
We assume that the eigenvalue $\lambda$ is simple, that $f$ has no roots and that $df$ has no roots.
This assures that $Z=0$ is a geometric $(d-1)$-graph. As $L=d d^*$ on $0$ forms, we get from $d^* d f = \lambda f$
that $d d^* d f = \lambda df$ so that $df$ is an eigenfunction to the one form Laplacian $L_1=d d^* + d^* d$. 
The set $\{ df=0 \}$ of all triangles and simplices where both $f$ and $df$ changes sign is the same than $f=0$.  \\

What is the topology of the hypersurface $Z_2 = \{ \vec{v}_2=0 \}$? For a $2$-sphere $G$, 
we know that there are two components so that the principal nodal curve $Z_2$ has to be a circle. 
Is the nodal curve to the second eigenvalue a $d$-sphere, if $G$ is a $d$-graph?
We believe that the answer is yes and robust. If $f$ should have roots, we can 
add a small random function with $|g(x)| \leq \epsilon$. We expect that for sufficiently small
$\epsilon>0$ and almost all $g$, the surface $Z = \{ f+g =0 \; \}$ has the same topology.

Does the topology of the nodal manifold to the second eigenvalue depend 
on the topology of $G$ only?  \\

We think the answer is yes as $Z$ has to be a connected surface and that going from a genus $k$
to a genus $k+1$ surface can not happen so easily. More generally, we expect:
if $G$ and $H$ are homotopic graphs of the same dimension, then the second nodal manifolds 
of $G$ and $H$ are homotopic. \\

\begin{figure}[ph]
\scalebox{0.2}{\includegraphics{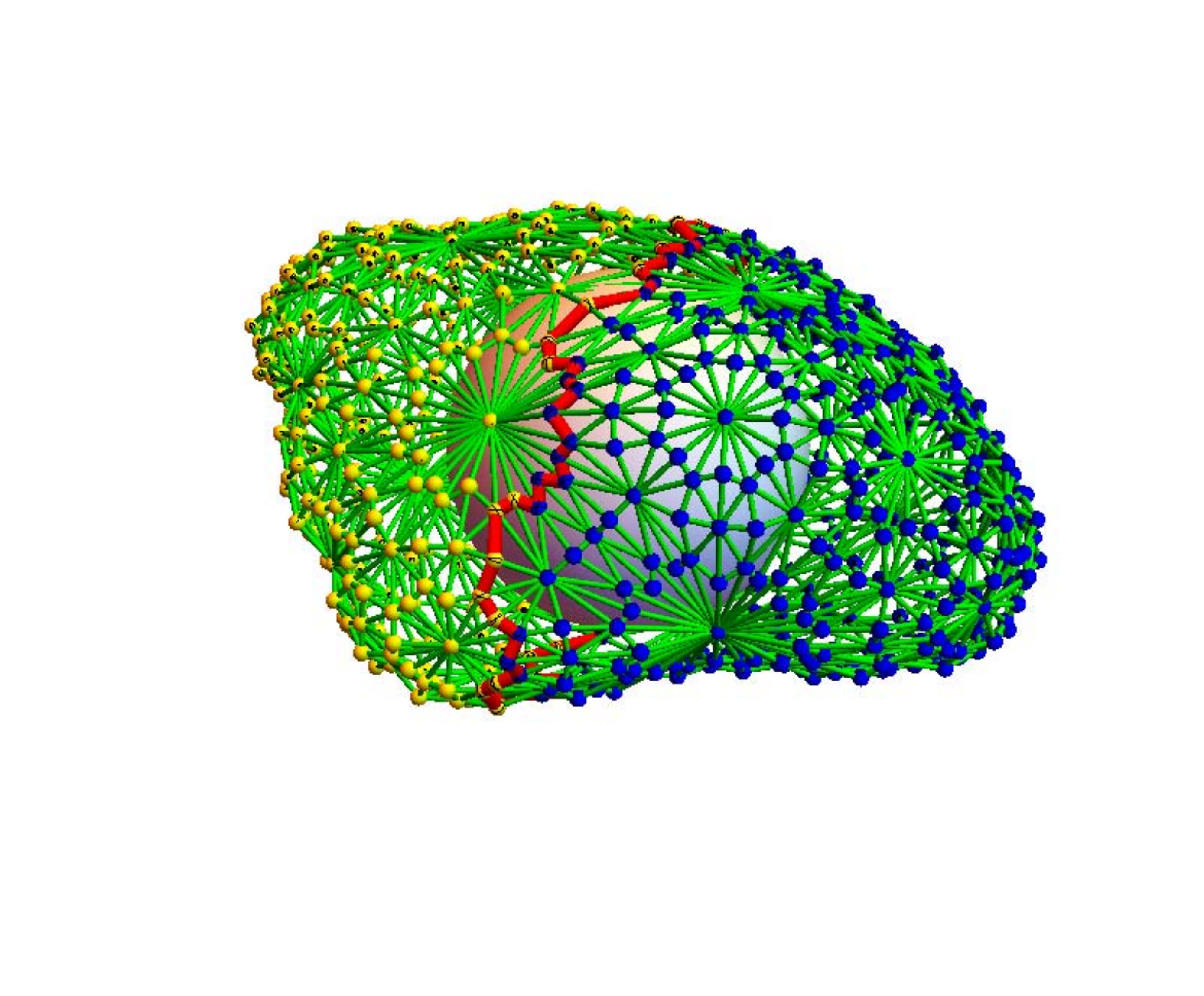}}
\scalebox{0.2}{\includegraphics{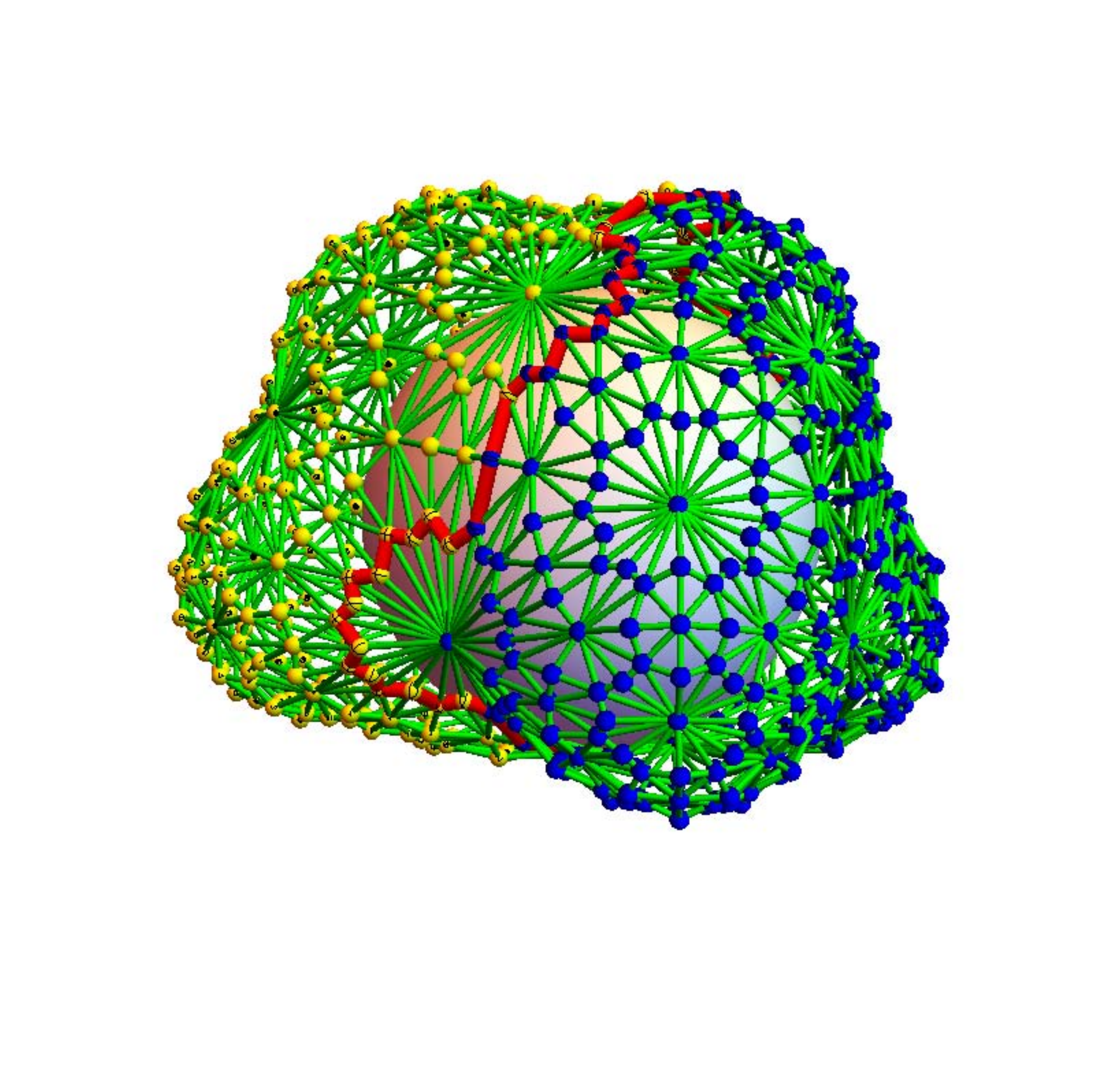}}
\scalebox{0.2}{\includegraphics{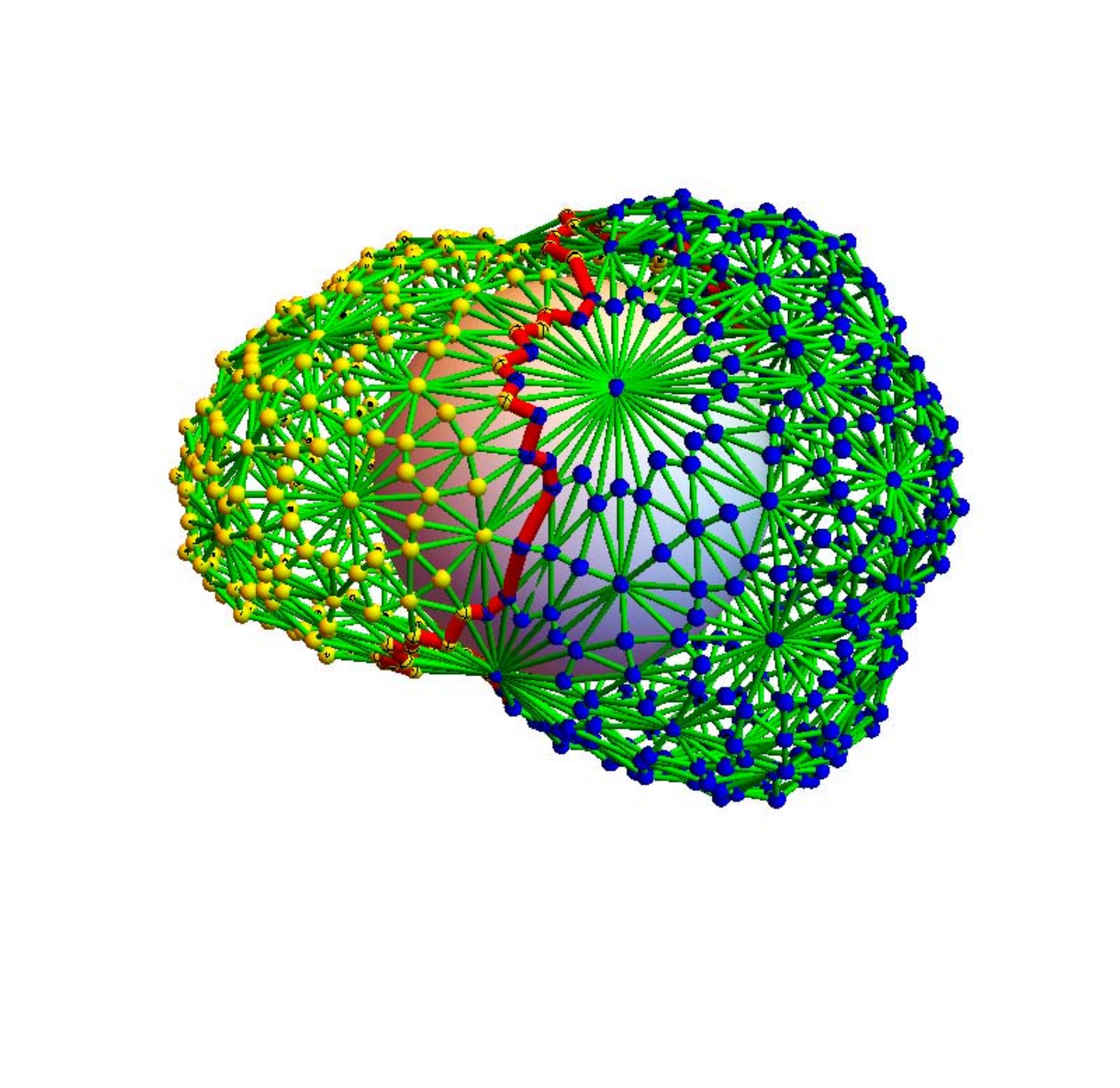}}
\scalebox{0.2}{\includegraphics{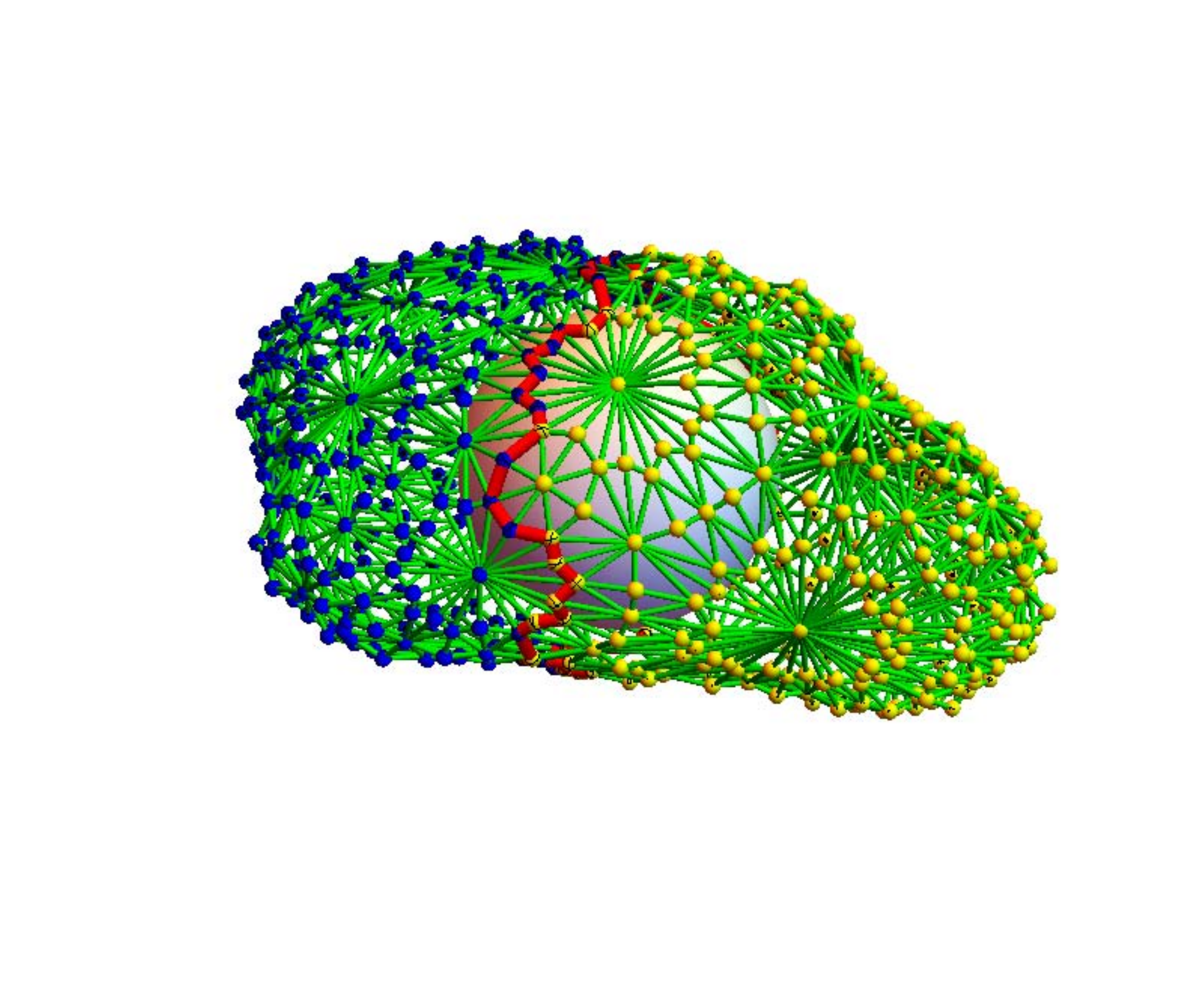}}
\scalebox{0.2}{\includegraphics{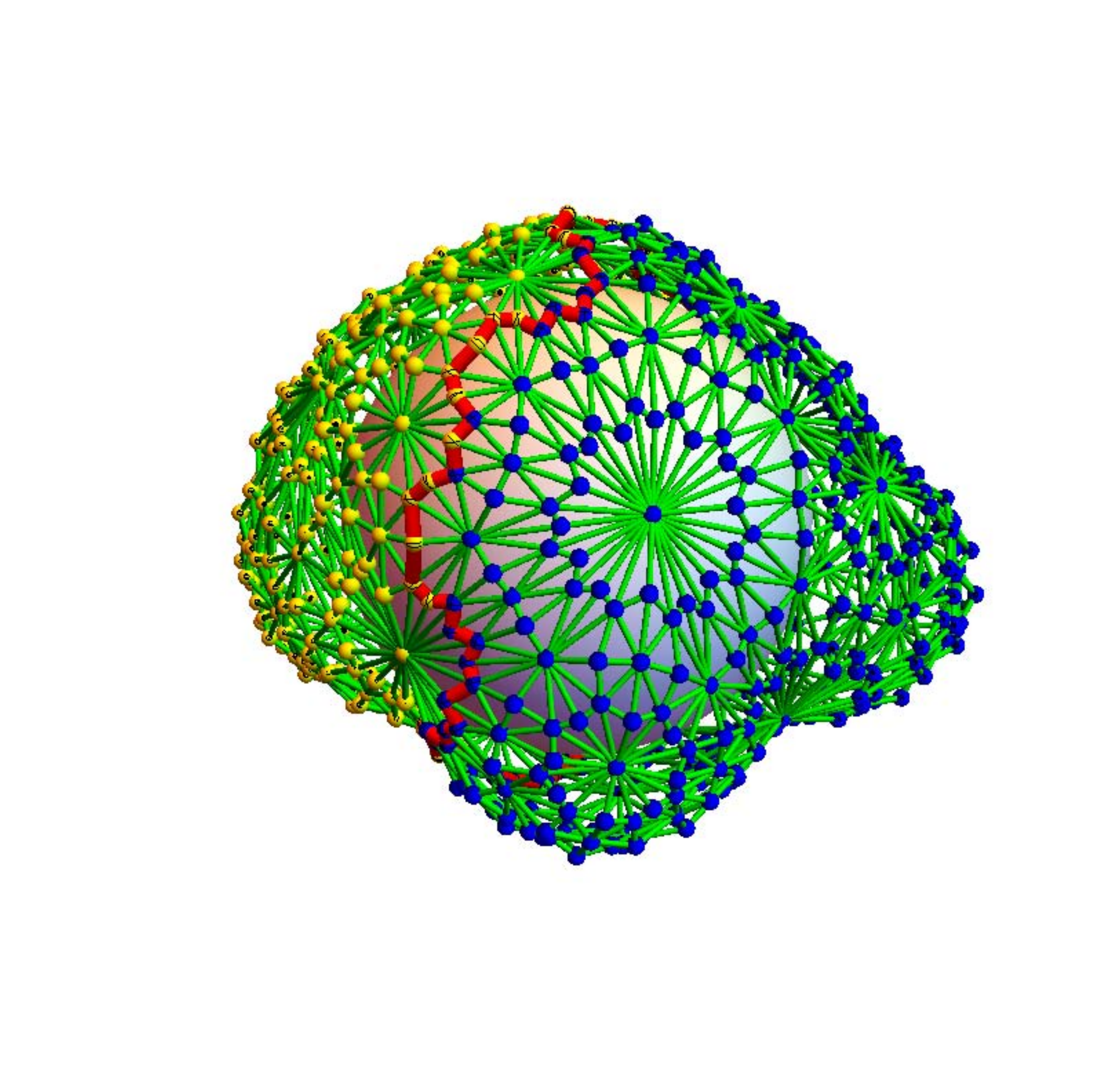}}
\scalebox{0.2}{\includegraphics{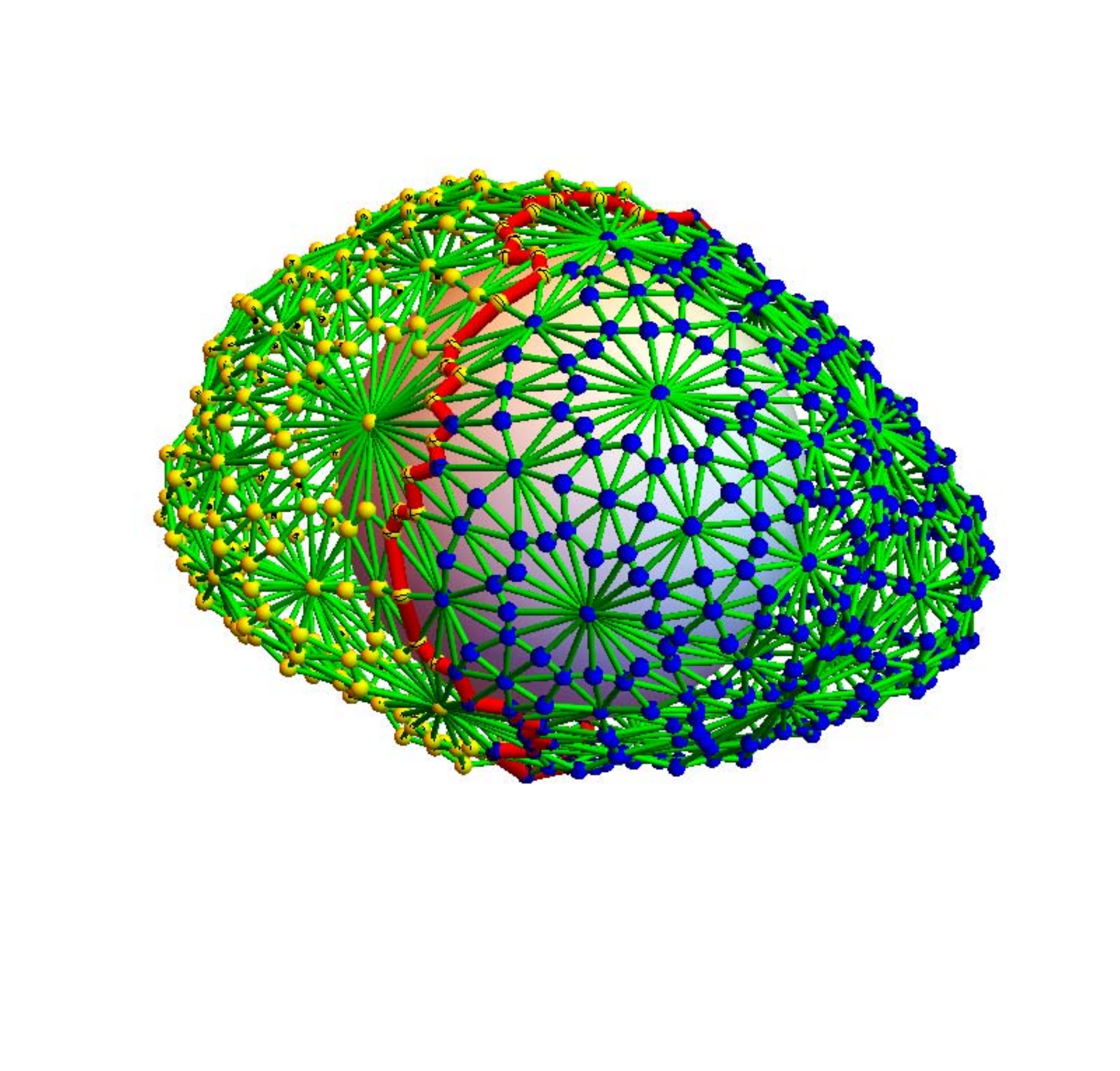}}
\caption{
Level surfaces of the second eigenfunction $f_2$ of the Laplacian on $G_2$, where
$G$ is a 2-sphere. The dividing surface is then always a Jordan curve. 
}
\end{figure}

\begin{figure}[ph]
\scalebox{0.2}{\includegraphics{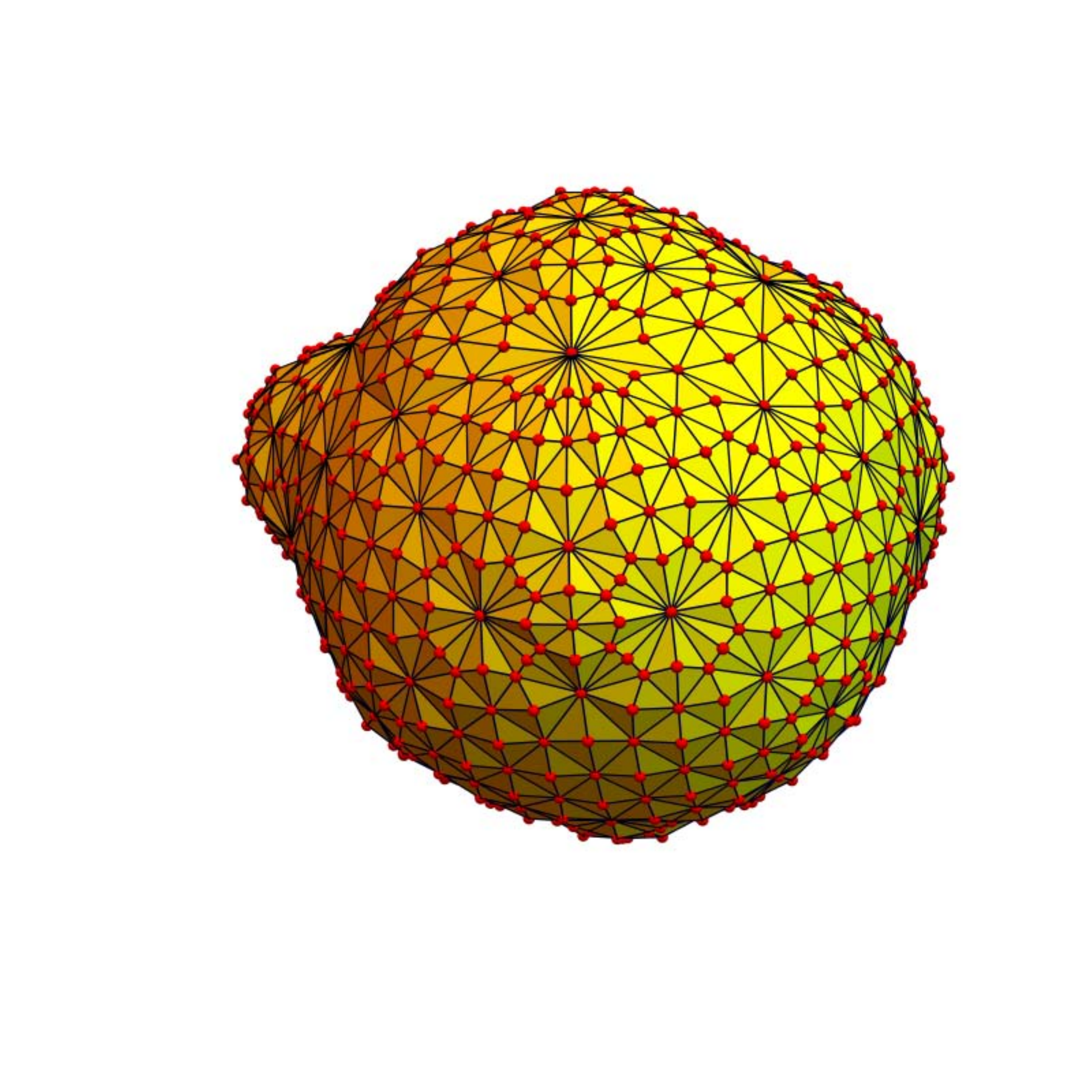}}
\scalebox{0.2}{\includegraphics{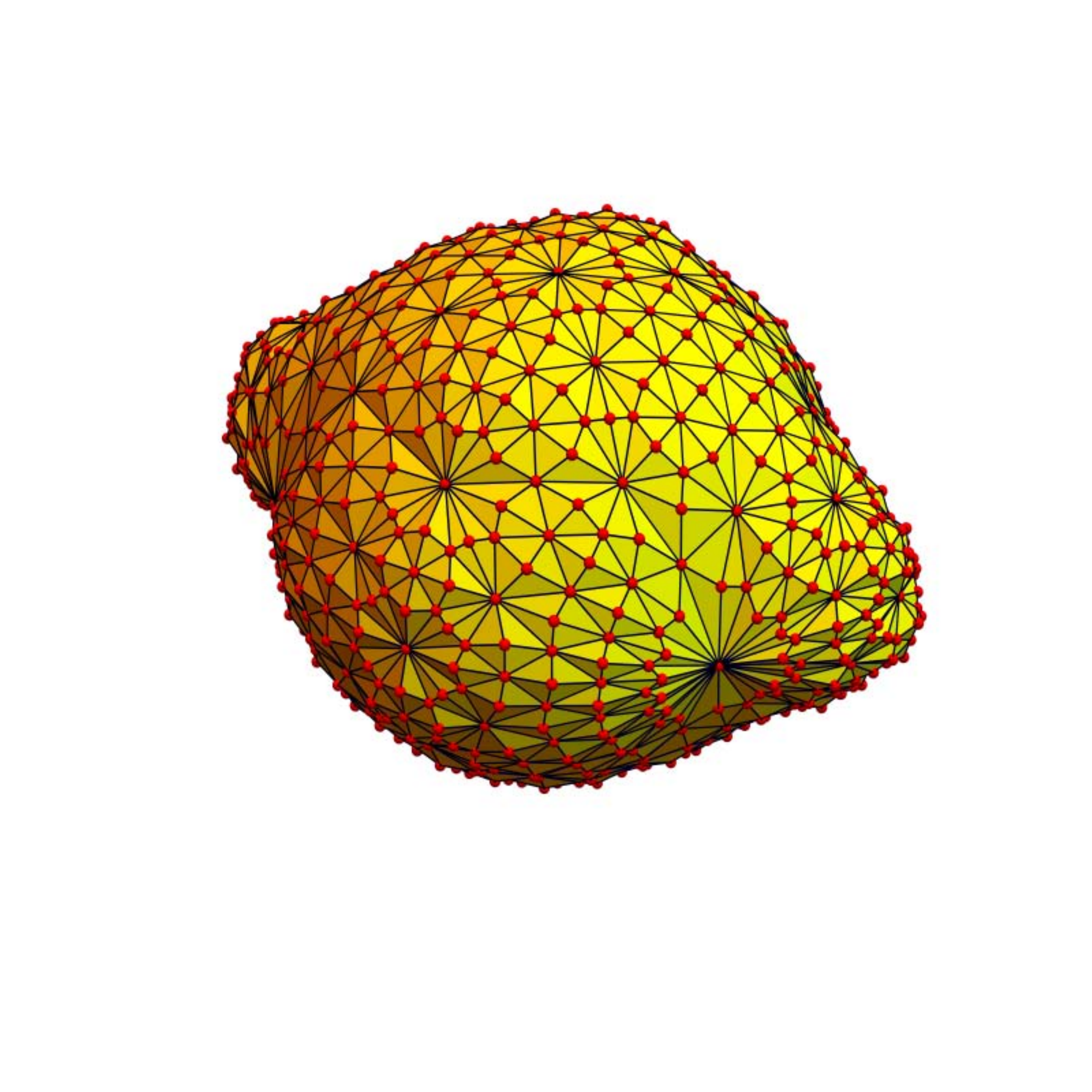}}
\scalebox{0.2}{\includegraphics{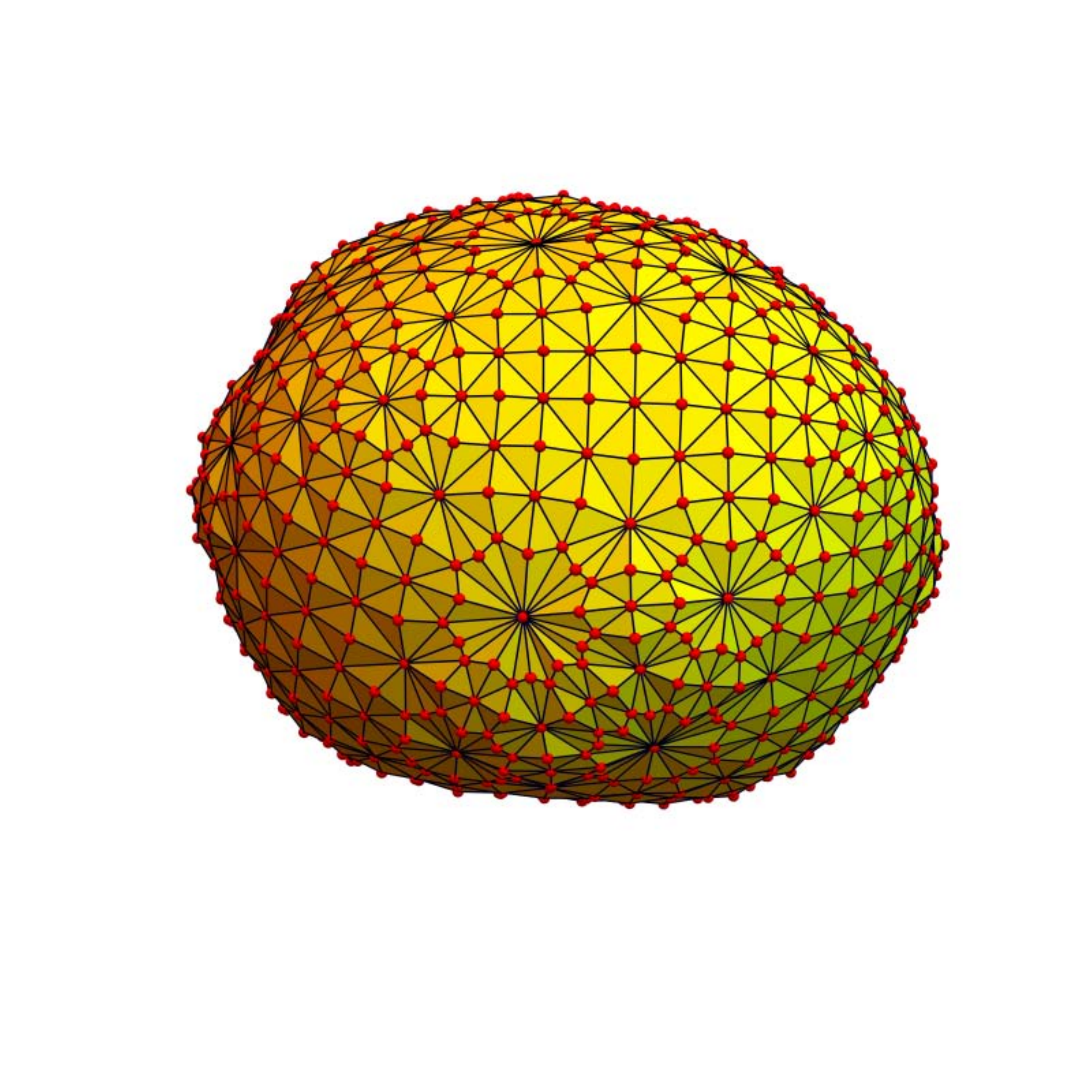}}
\scalebox{0.2}{\includegraphics{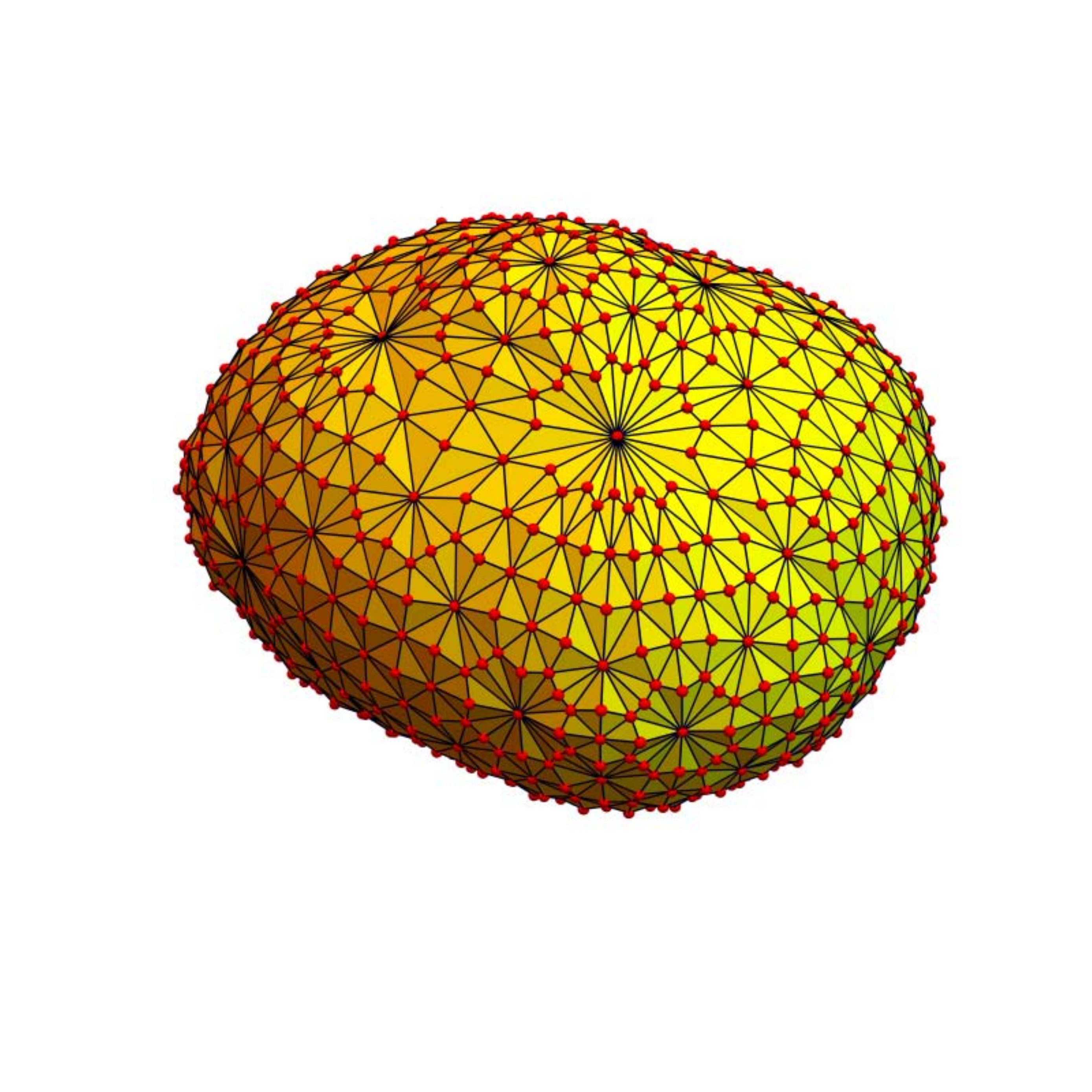}}
\scalebox{0.2}{\includegraphics{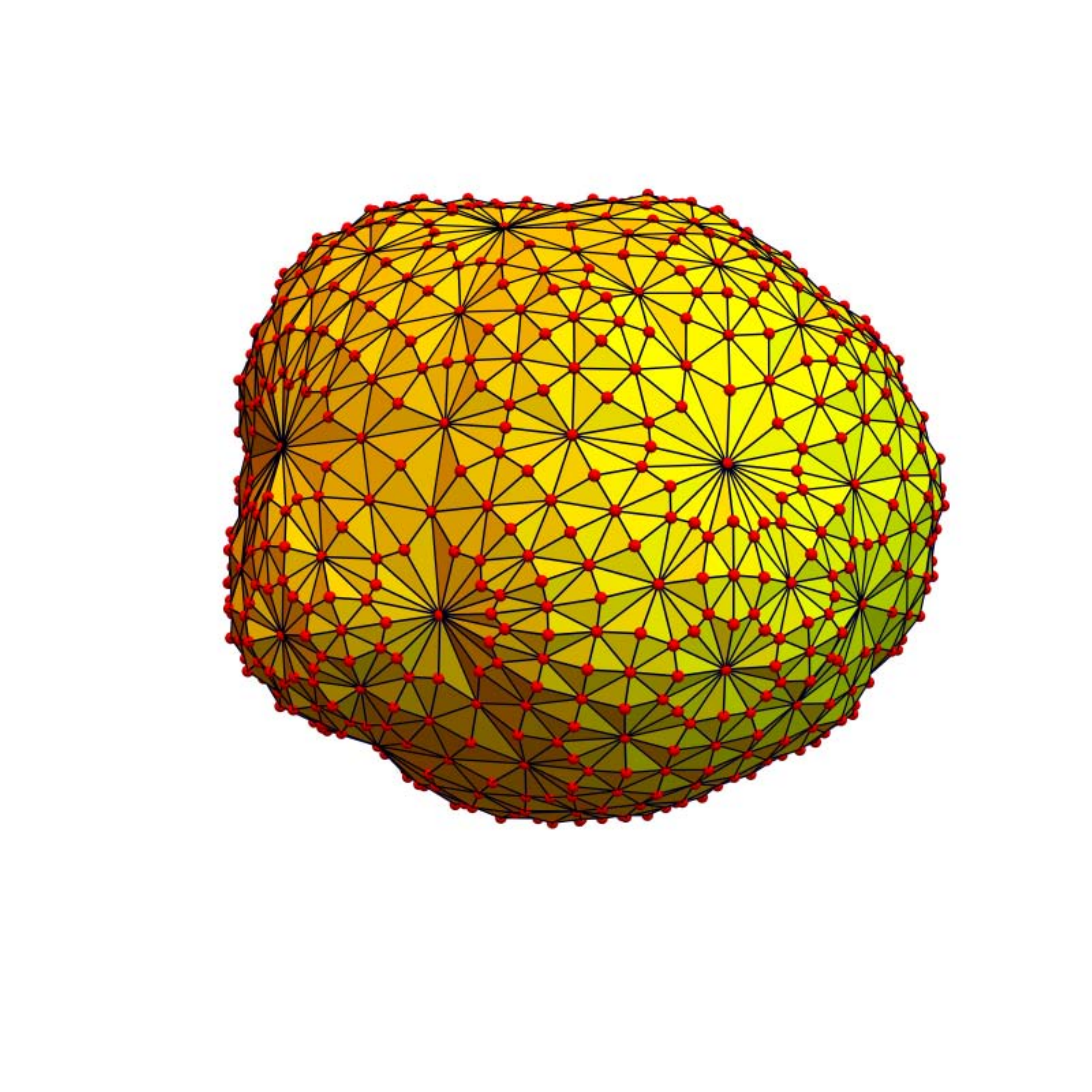}}
\scalebox{0.2}{\includegraphics{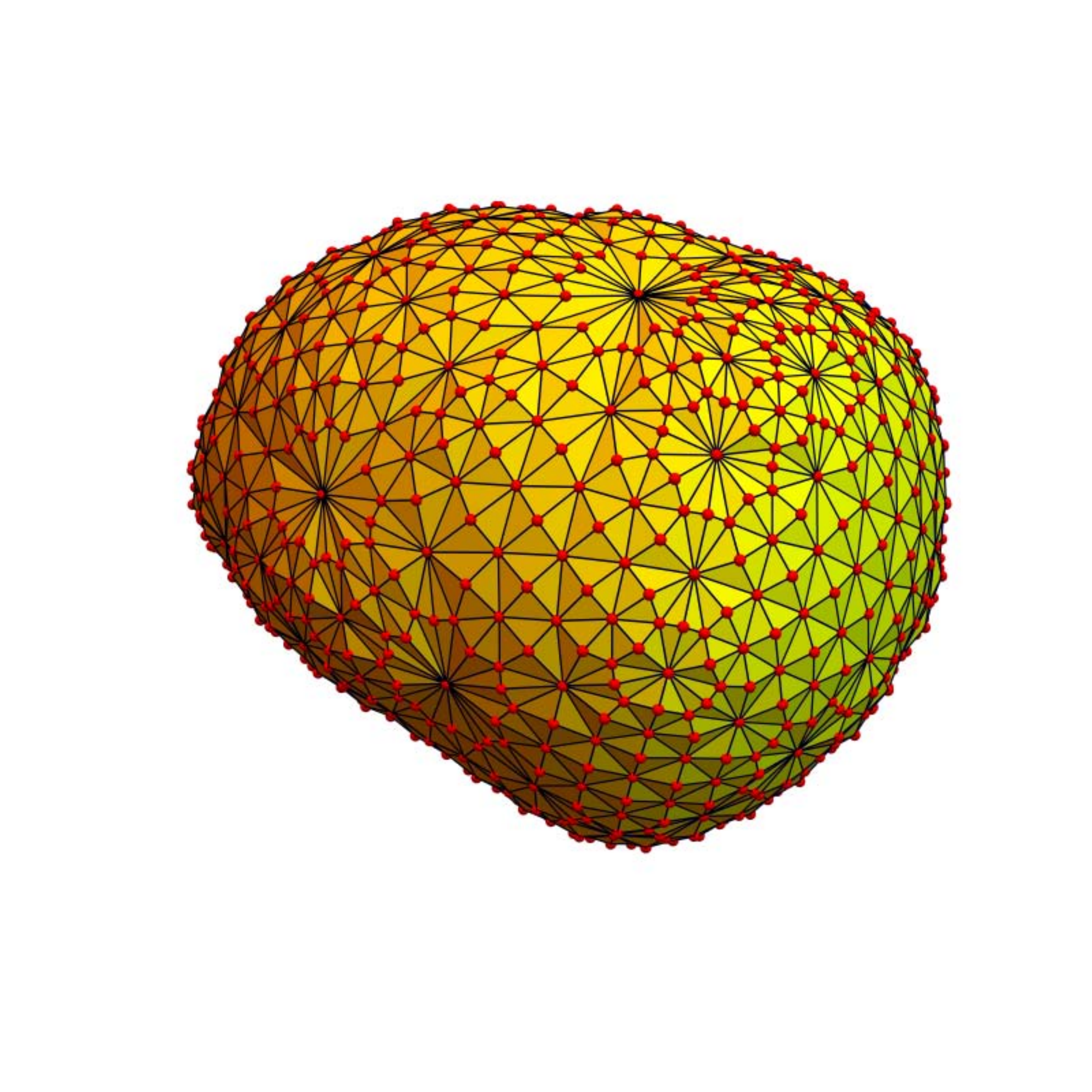}}
\caption{
Level surfaces of the second eigenfunction of the Laplacian,
where $G$ are 3-spheres. The dividing surface is then a $2$-graph.
We believe that $\{ f_2=0 \}$  is always a 2-sphere, if $G$ is a
$3$-sphere and $f_2$ is the eigenvector to the smallest nonzero eigenvalue.
A more general question is whether two whether the nodal region $\{ f_2>0 \}$ 
is always simply connected if $G$ is a $d$-sphere. 
}
\end{figure}

\begin{figure}[ph]
\scalebox{0.2}{\includegraphics{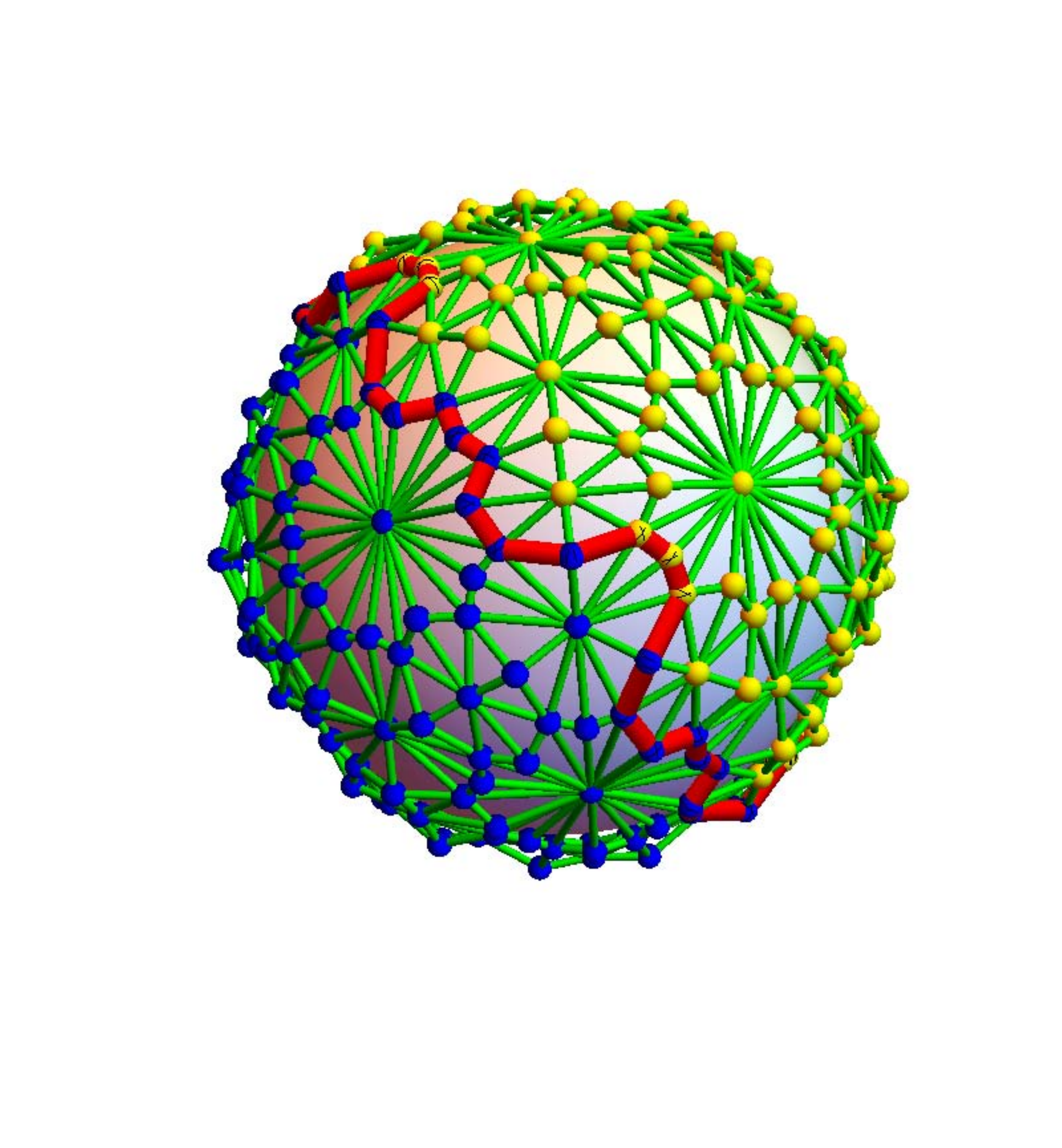}}
\scalebox{0.2}{\includegraphics{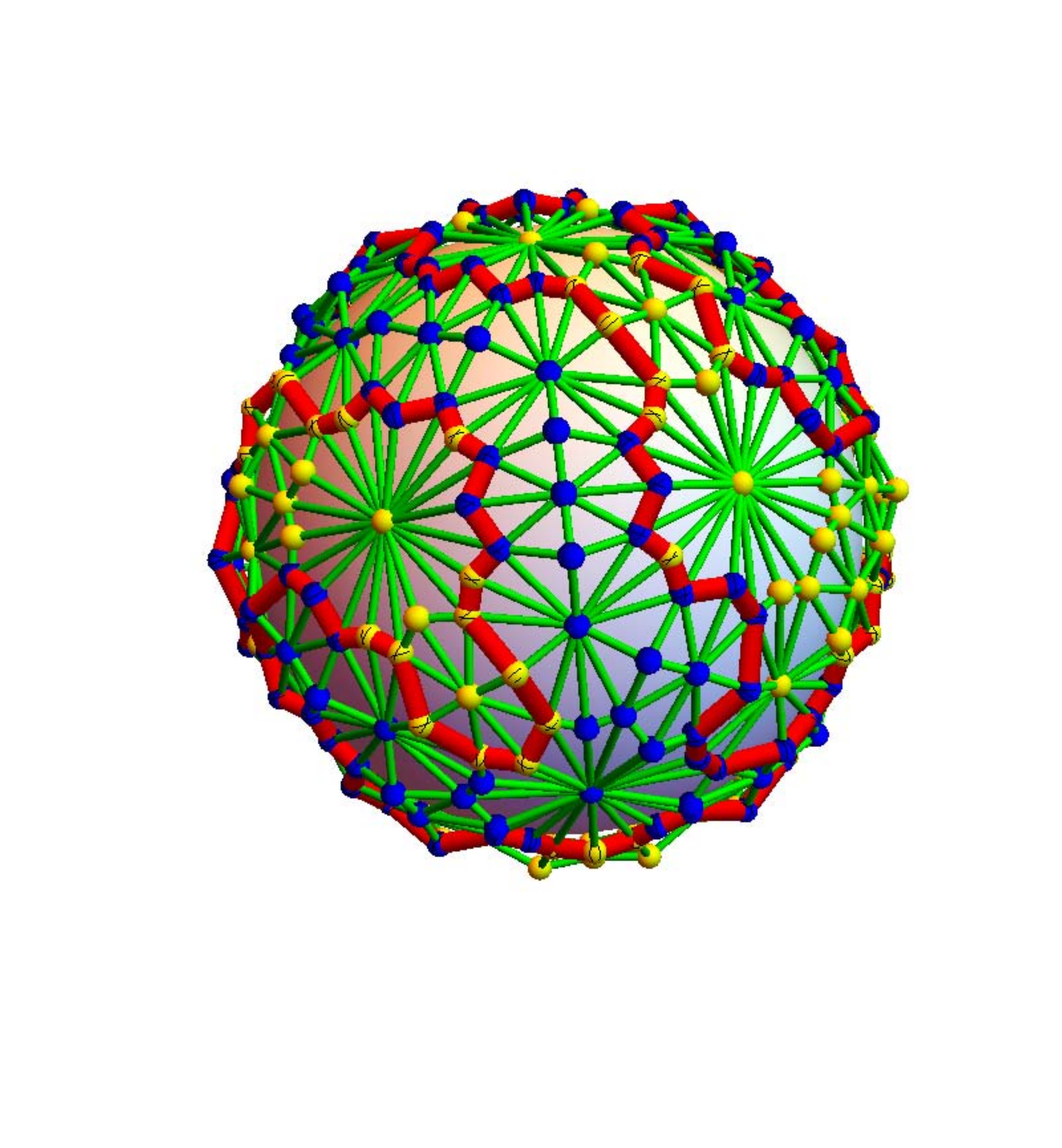}}
\scalebox{0.2}{\includegraphics{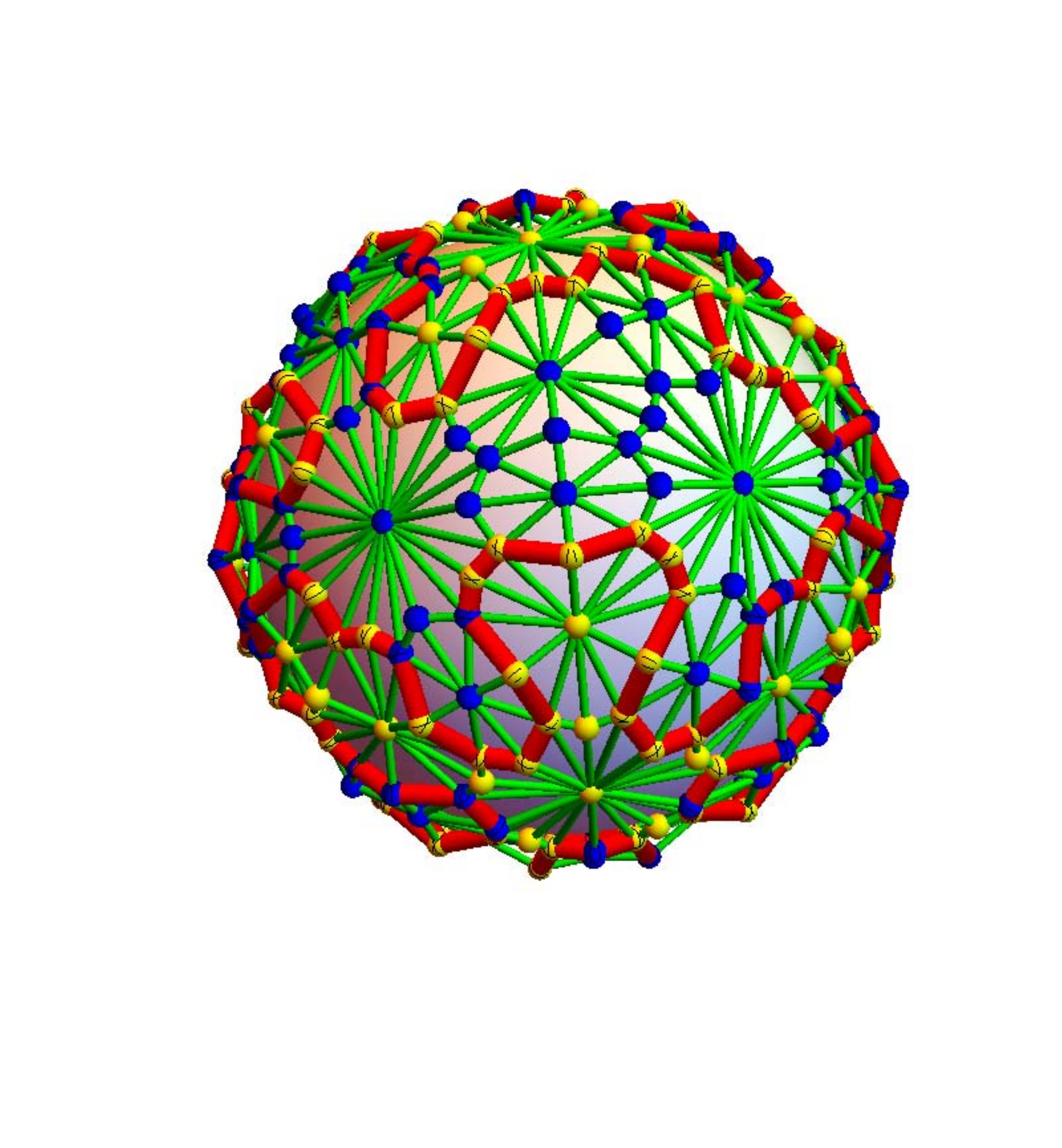}}
\scalebox{0.2}{\includegraphics{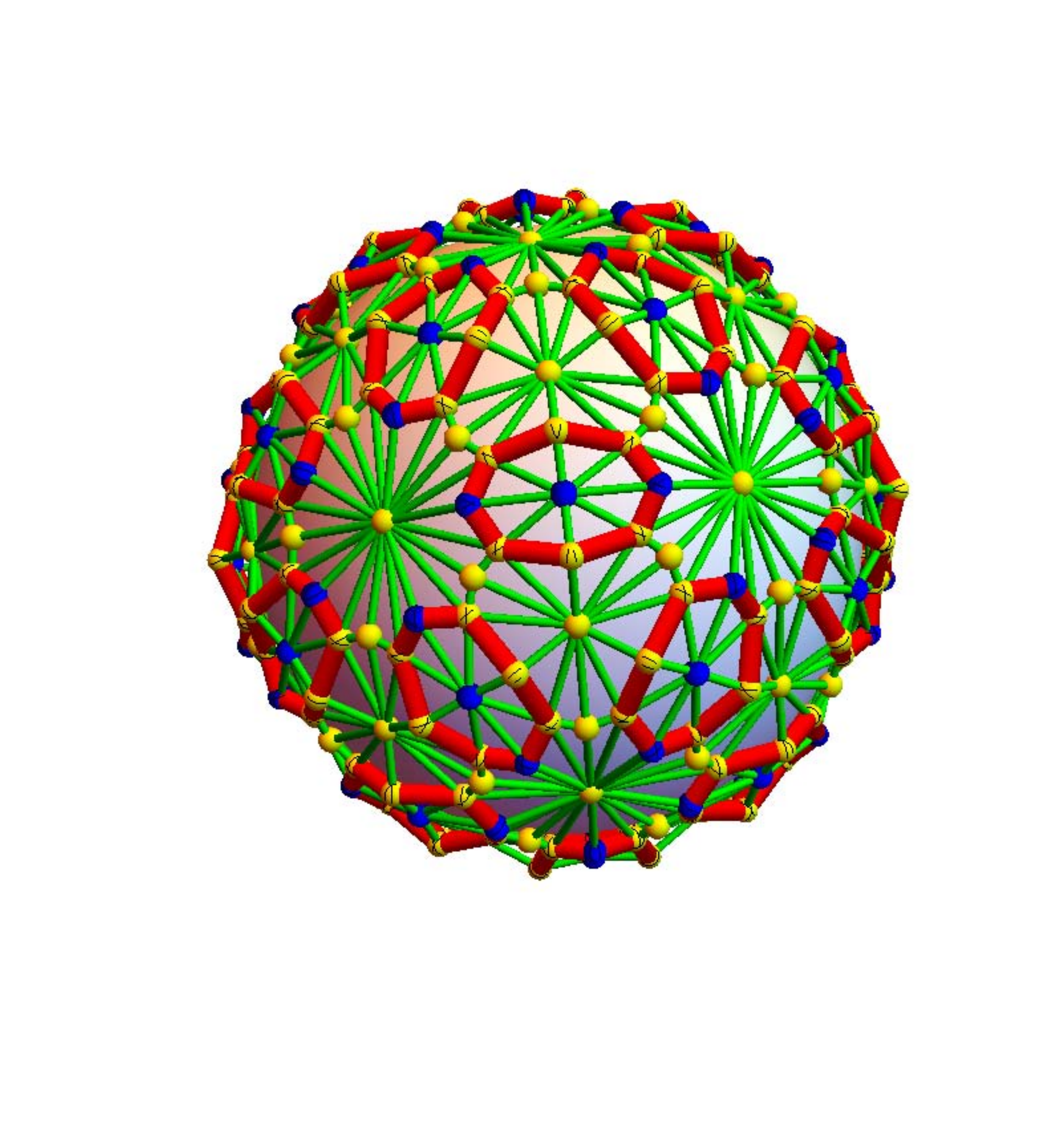}}
\scalebox{0.2}{\includegraphics{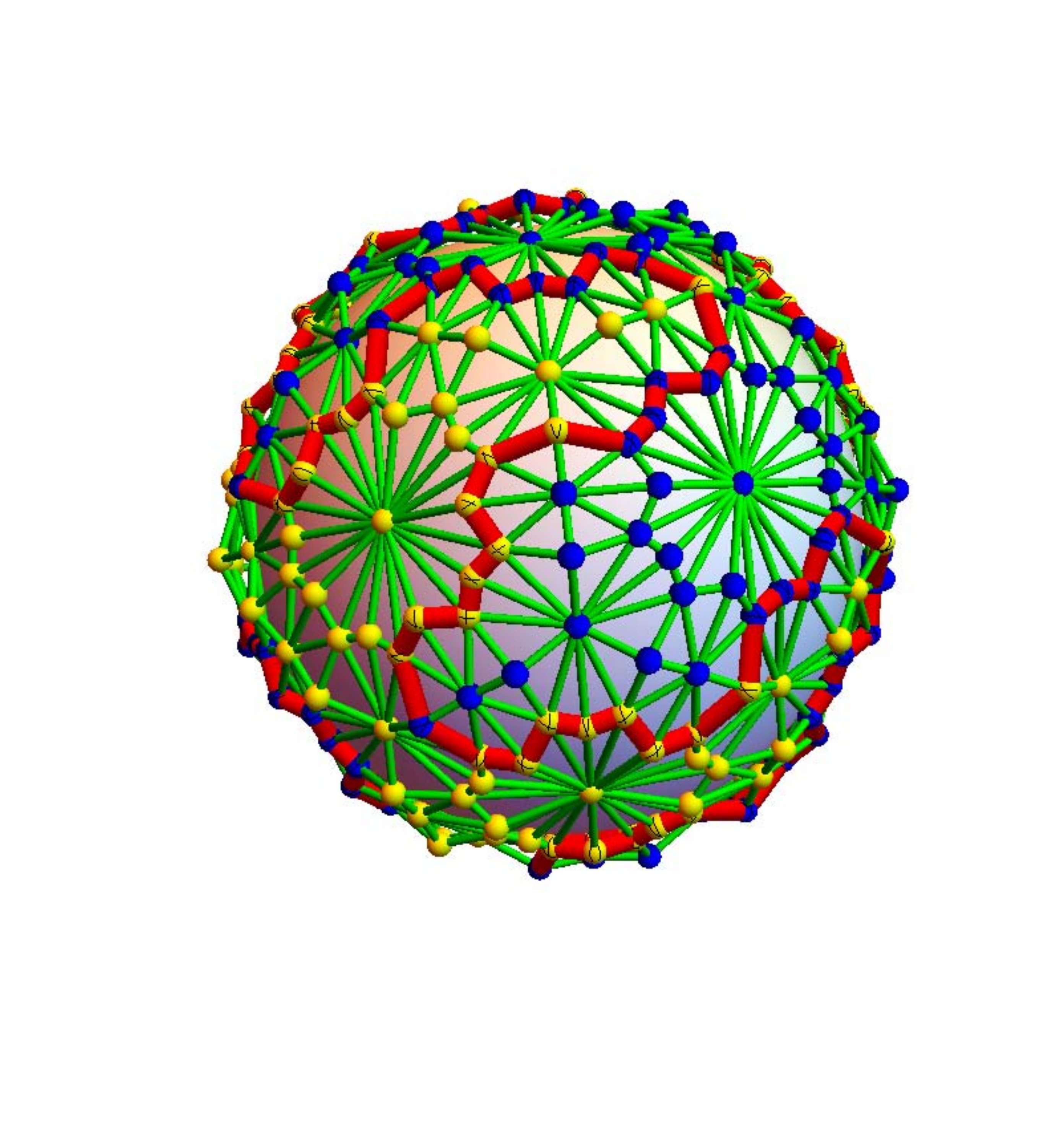}}
\scalebox{0.2}{\includegraphics{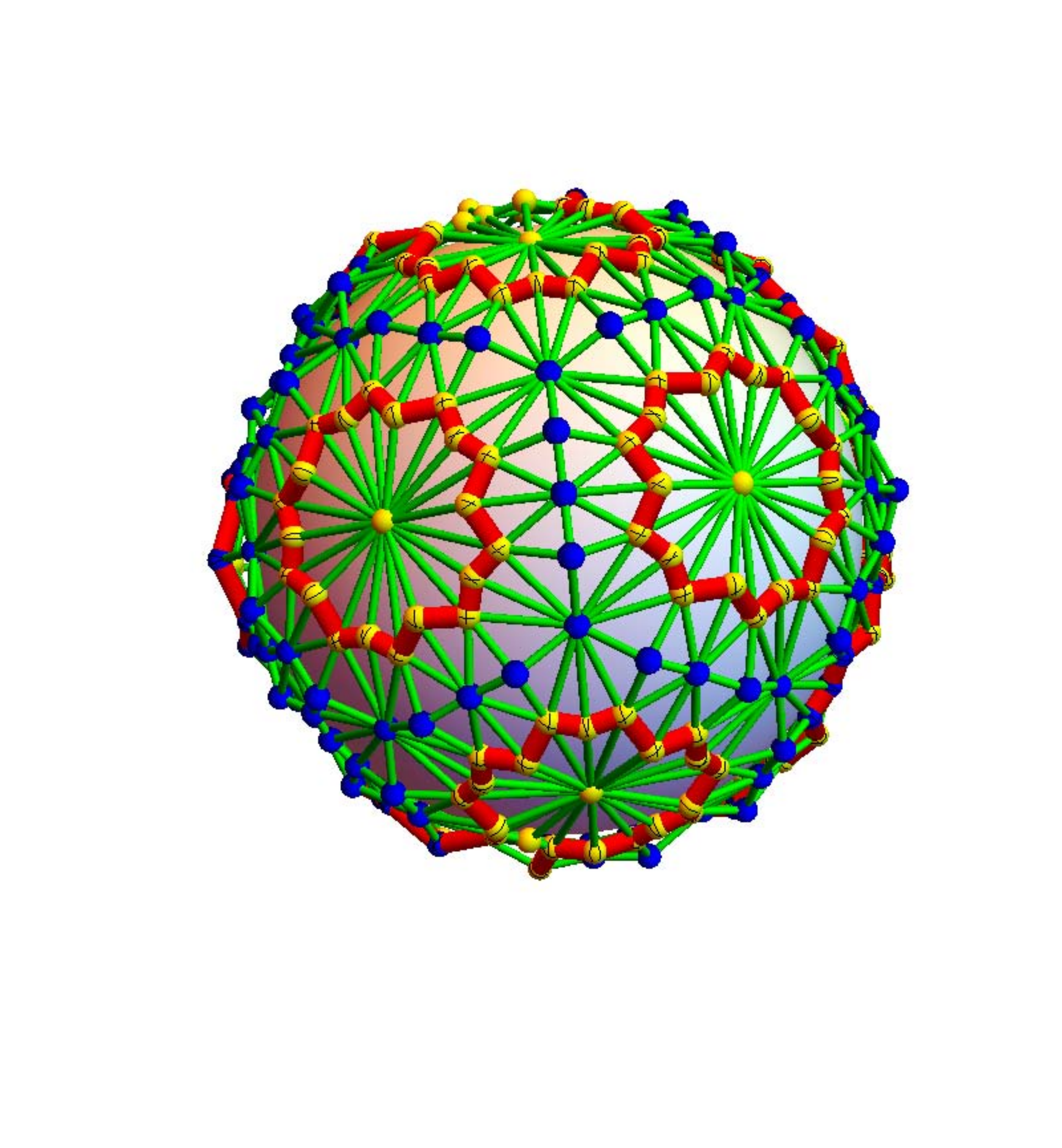}}
\caption{
Level surfaces of eigenfunctions of the Laplacian on a graph $G_1$,
where $G$ is the icosahedron. The Courant-Fiedler nodal theorem tells
that the number of positive nodal regions of $\lambda_n$ is $\leq n$.
It is confirmed in the pictures seen. The first picture shows the second
eigenvector. The first one is constant and has only one region.
}
\end{figure}

\begin{figure}[ph]
\scalebox{0.2}{\includegraphics{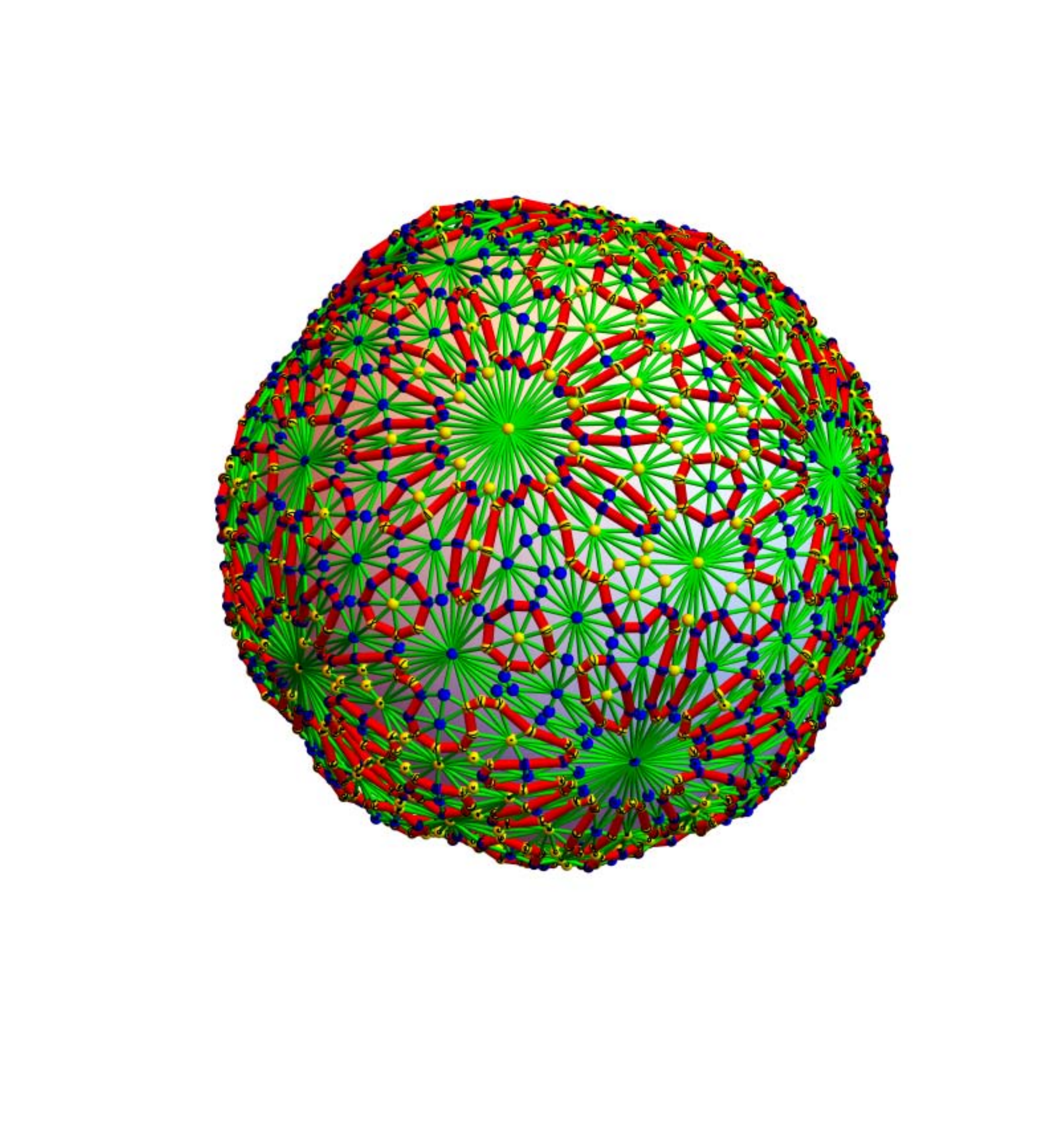}}
\scalebox{0.2}{\includegraphics{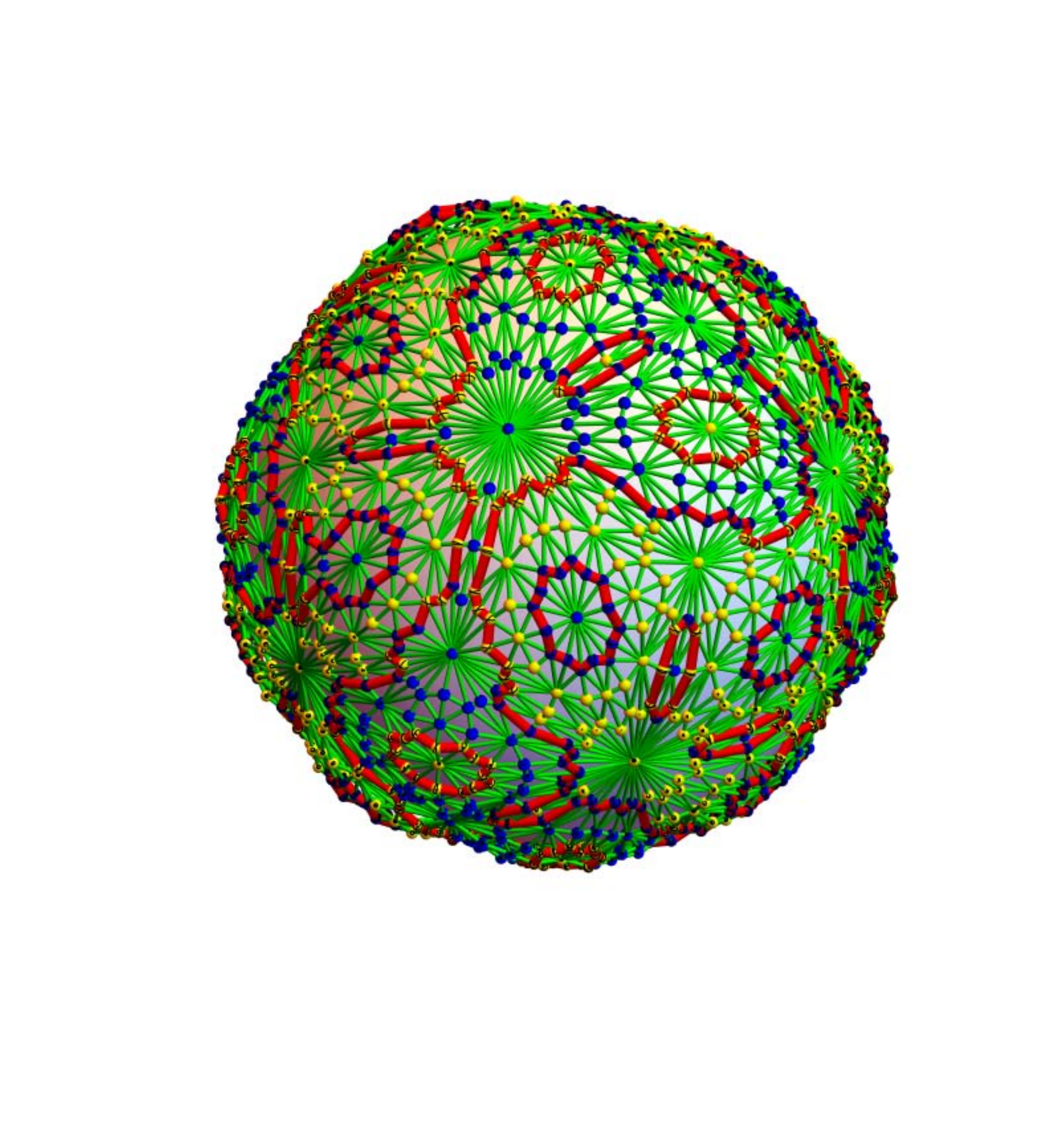}}
\scalebox{0.2}{\includegraphics{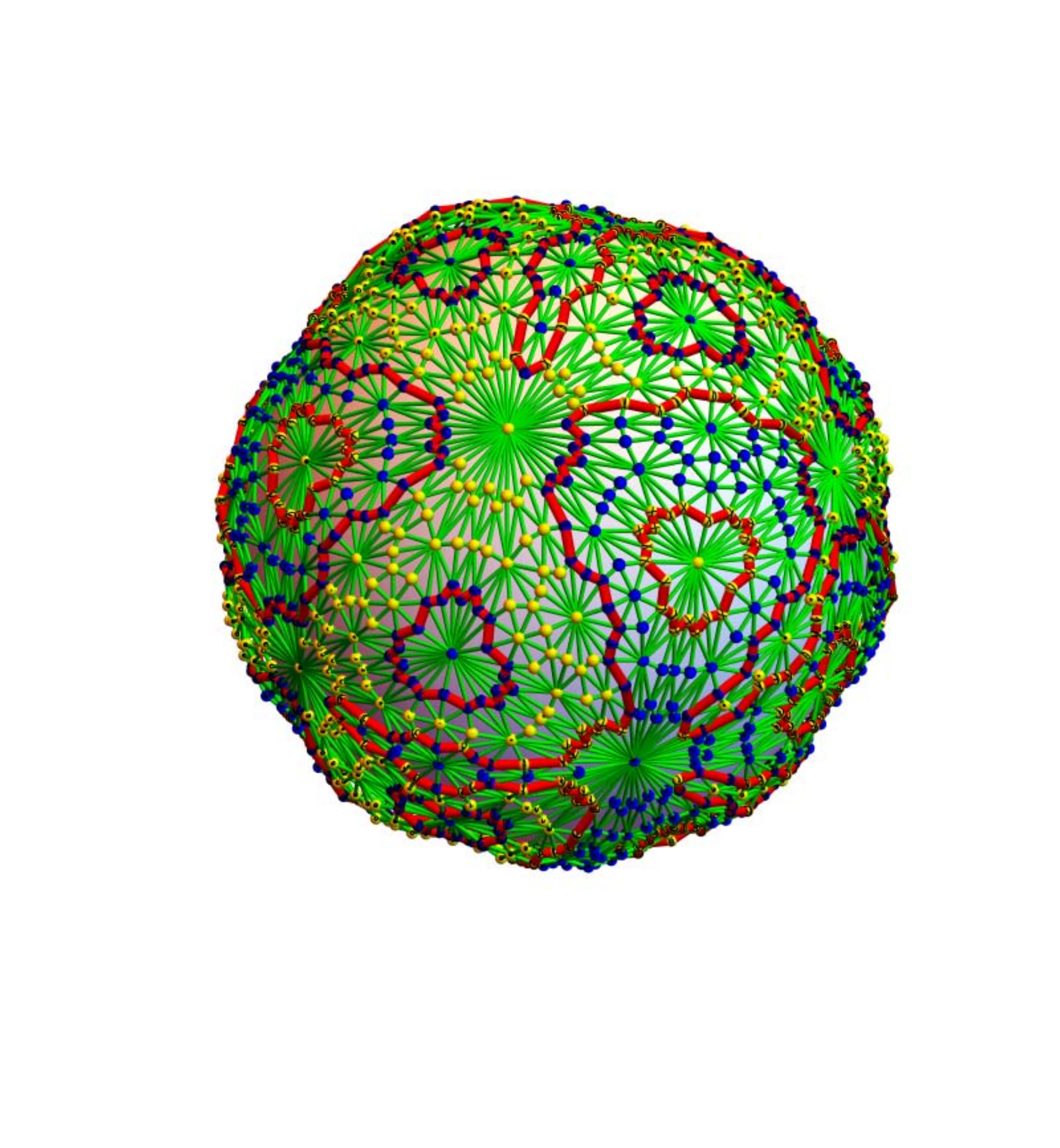}}
\scalebox{0.2}{\includegraphics{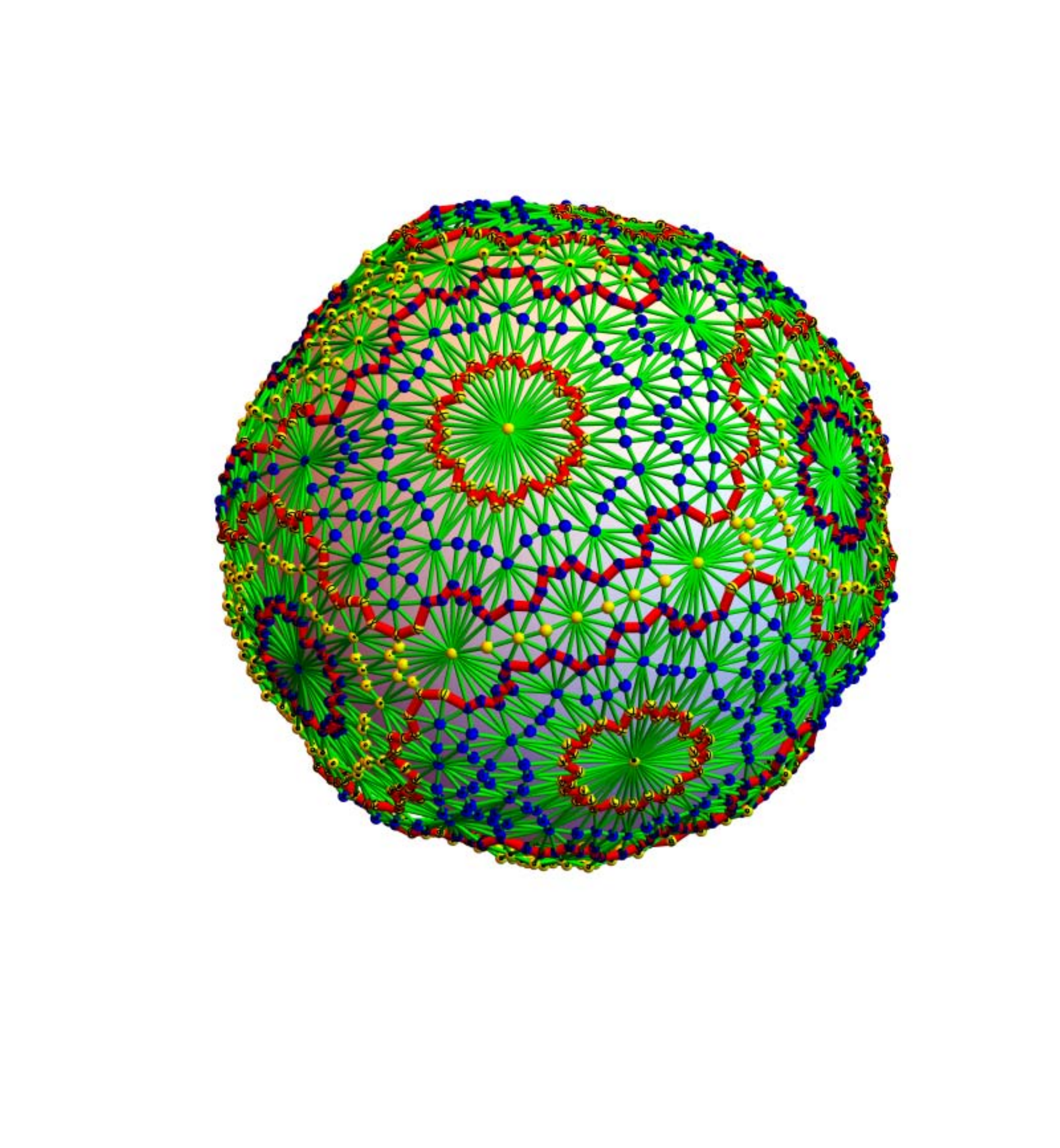}}
\scalebox{0.2}{\includegraphics{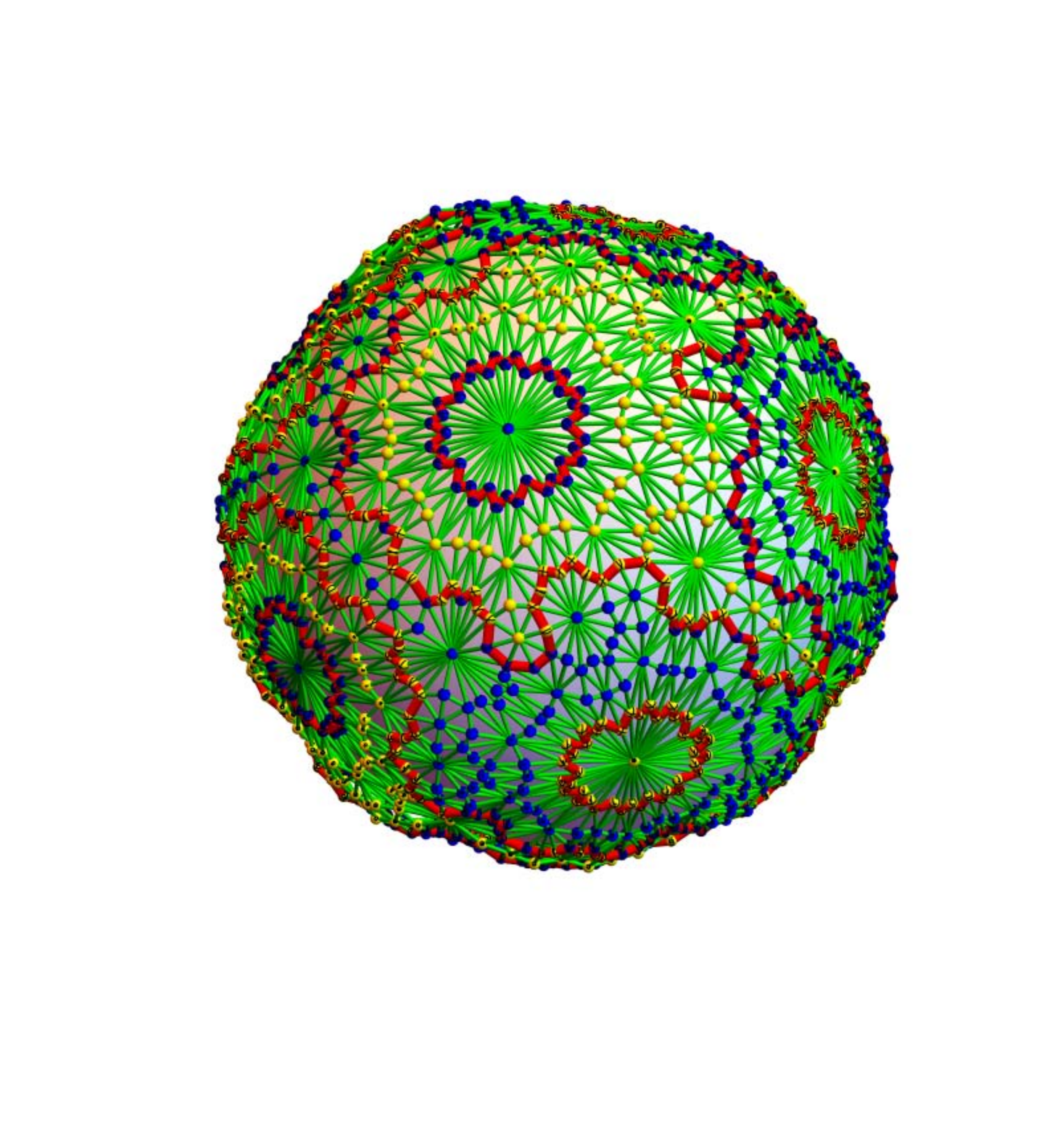}}
\scalebox{0.2}{\includegraphics{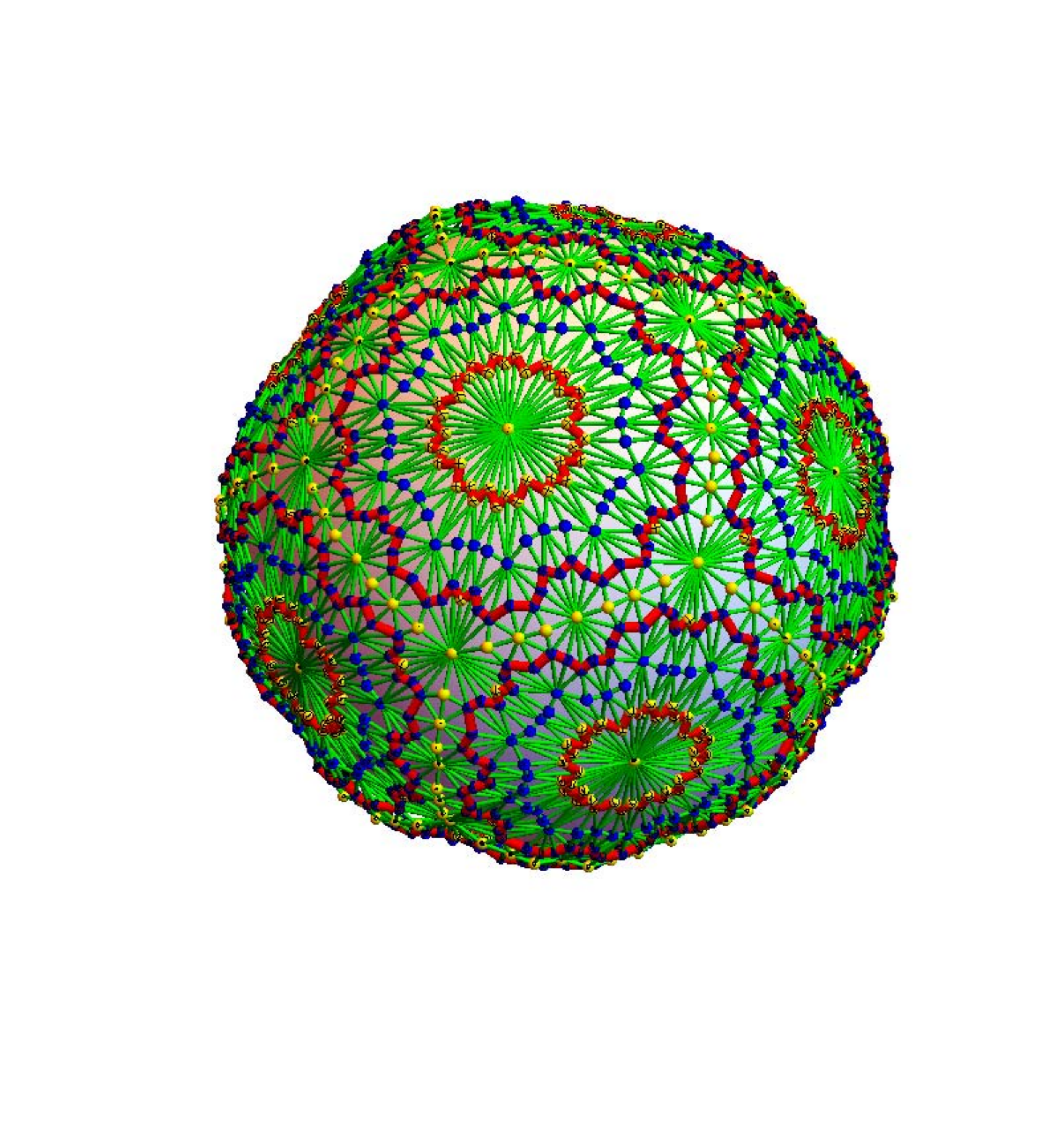}}
\caption{
Level surfaces of eigenfunctions of the Laplacian on $G_2$, where
$G$ is the icosahedron. Also this illustrates the Courant-Fielder theorem.
The first picture shows the second
eigenvector for which the surface $Z_2$ divides the graph in to two regions. 
The first one is constant and has only one region. 
}
\end{figure}

\section{Summary}

A $d$-graph is a finite simple graph for which every unit sphere is a $(d-1)$ sphere. 
Given $k$ locally injective real-valued functions $f_1,\dots,f_k$ on the
vertex set of a $d$-graph, we can define the zero locus $Z=\{\vec{F}=\vec{0}\}$ of $F$
in $G$ as the graph $(V,E)$ with vertex set $V$ consisting of complete subgraphs in $G$, on 
which all the functions $f_j$ change sign and for which two vertices are connected 
if one is a subgraph of the other. The graph $Z$ is a subgraph of the barycentric refinement $G_1$
but like varieties in the continuum, the graph might be singular and not be a $(d-k)$-graph.
We can define a discrete tangent bundle on $G_1$ so that the projection of $\nabla f$ 
in $T_xG = Z_2^d$ to $T_xH$ is zero, establishing 
the graph analogue of ``gradients are perpendicular to level surfaces". 
If every unit sphere $S(x)$ in $G_1$ has the property that 
$Z \cap S(x)$ is a $(d-k-1)$-sphere in the $(d-1)$-sphere $S(x)$, 
then $F=c$ is a $(d-k)$-graph $H$. The later plays the role of 
$\{f_1=c_1, \dots =f_k=c_k \}$ for $k$ differentiable functions $f$
on a smooth manifold $M$ for which the maximal rank condition
${\rm rank}(dF)(x)=d-k$ for $F=(f_1,\dots,f_k), x \in Z$ is satisfied.
Discrete Lagrange equations help to maximize or minimize $f$ under a constraint $g=c$: as in the continuum,
critical points need parallel gradients $\nabla f, \nabla g$. \\

If the {\bf ordered zero locus} is defined by looking at hypersurfaces on 
successive barycentric refinements, we never run into singularities if the 
functions are strongly locally injective in the sense that the 
function values of $f_j$ on the vertex set are 
rationally independent on each complete subgraph.
The {\bf Sard theorem} states then that the graph $\{f_1=c_1,\dots,f_{k}=c_{k} \}$ is a $(d-k)$-graph inside
the $d$-graph $G_k$ for all $\vec{c}$ not in a finite set. An application is the observation
that any $d$-dimensional projective algebraic set admits an approximation by manifolds with
a triangulation which is a minimally $(d+1)$-colorable graph,
using an explicitly constructed $d$-graph determined by the equations defining $V$. \\

The possibility to define level surfaces in discrete setups can be illustrated in the case of
eigenfunctions of the Laplacian. Pioneered by Chladni, the geometry of nodal surfaces is important
for the physics of networks. Of special interest is the {\bf ground state} $f_2$ of 
a compact Riemannian manifolds or finite graph. The nodal surface $f_2=0$ is a $(d-1)$-dimensional 
hypersurface and the double nodal surface $f_2=0,f_3=0$ is generically a codimension $2$ manifold. 
If the manifold or graph is a 3-sphere, the nodal surface is two dimensional and the double nodal surface 
is a collection of closed curves in $S^3$. We asked here whether $f_2=0$ have positive genus and whether
$f_2=0,f_3=0$ can be knotted. Graph theory allows to investigate this experimentally. The answers 
for Riemannian manifolds and graphs are expected to be the same. \\

Lets look at some {\bf multi-variable calculus terminology} in a $2$-graph
translated to graph theory: \\

\begin{tabular}{ll} \hline
Critical point $x$      &    $S^-_f(x)$ is not contractible \\
Discriminant $D$        &    Poincar\'e-Hopf index $i_f(x)=1-\chi(S^-_f(x))$ \\
$D>0,f_{xx}<0$          &    $i_f(x)=1$, $S^-_f(x) = S_f(x)$ \\
$D>0,f_{xx}>0$          &    $i_f(x)=1$, $S^-_f(x) = \emptyset$ \\
$D<0$                   &    $i_f(x)<0$                  \\
$D=0$                   &    not locally injective  \\
Level curve $f(x,y)=0$  &    zero locus $f=0$ in $G_1$. \\
$T_xM$                  &    maximal simplex $t$ in $G$ containing $x$ \\
Tangent vector $T_xM$   &    $\nabla f = {\rm sign}(f(y)-f(x)) \; | \; y \in t \}$  \\
Lagrange equations      &    $\nabla f = \lambda \nabla g$ or $\nabla g=(0,0)$\\ \hline
\end{tabular}

\begin{figure}[ph]
\scalebox{0.25}{\includegraphics{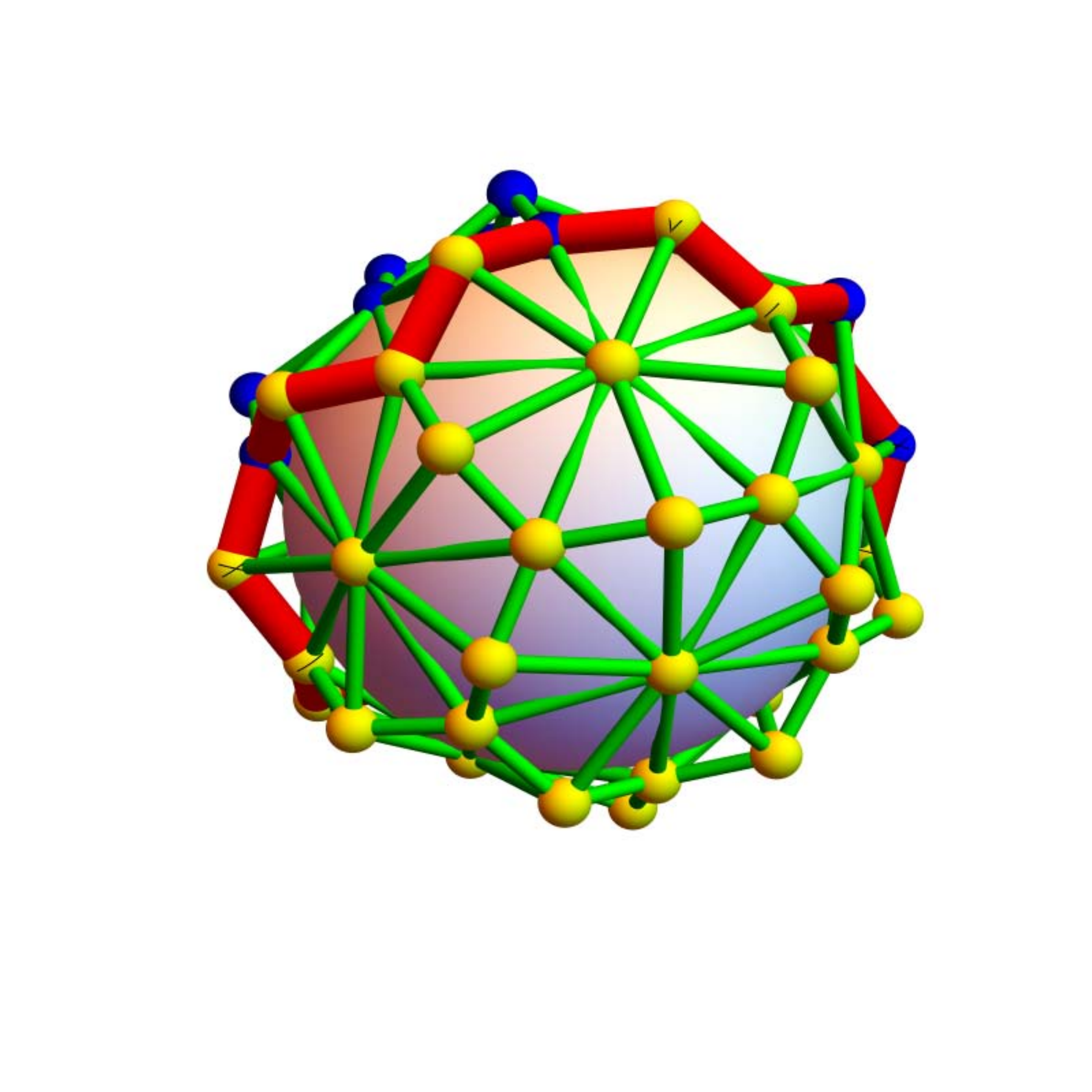}}
\caption{
The nodal hypersurface $\{ f=0 \}$ of the ground state. The length of the
surface is $20$, the two nodal domains have $30$ triangles each. The Cheeger number
(see \cite{Chung2007} for a discrete version), $c(f) = |C(f)|/{\rm min}(|A(f)| |B(f)|)$ 
satsifies the Cheeger inequality. 
}
\end{figure}

\begin{figure}[ph]
\scalebox{0.25}{\includegraphics{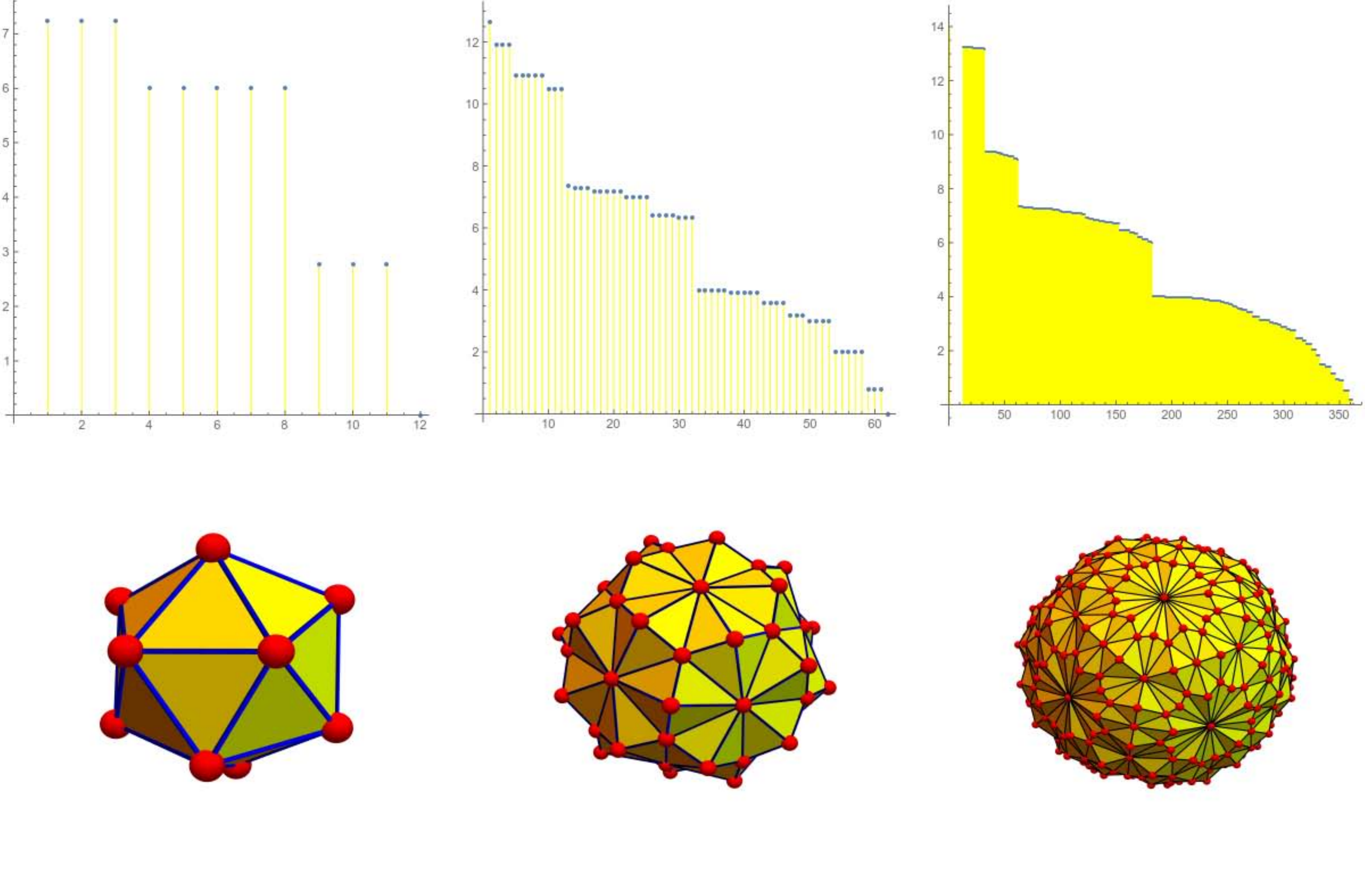}}
\caption{
The spectrum of refinements $G,G_1,G_2,G_3$  of the icosahedron graph $G$.
We noticed that the spectral integrated density of states
converges to a function which only depends on the dimension of the maximal
complete subgraph. While for graphs
without triangles, the convergence of $F_G(x) = \lambda_{[k x]}$
to the limiting function $4 \sin^2(\pi x)$ is immediate, in higher dimensions,
the limiting function appears to be non smooth. See \cite{KnillBarycentric}.
}
\end{figure}

\section{Questions}

A) In the commutative case, where $F=c$ is a subgraph of the barycentric refinement $G_1$,
the set of critical values can have positive measure.
It would be nice to find an upper bound on the measure of constants $c$ for which 
the graph $\{ F=c \} \subset G_1$ can be singular.  \\
B) In the context of graph colorings, the following question is analog to the
{\bf Nash embedding problem:}
Is every $d$-graph a subgraph of a barycentric refinement $G_n$ of some geometric $d$-graph? 
If the answer were yes, we would have a bound $d+1$ for the chromatic number of $G$.
A special case $d=3$ would prove the $4$-color theorem. We approached 
this problem by writing the sphere as embedded in a
$3$ sphere then making homotopy deformations to render the sphere Eulerian and so $4$ colorable
coloring in turn the embedded sphere.
While a Whitney embedding of a graph is possible if we aim for a homeomorphic image, we don't have
yet a sharp discrete analogue of the classical Whitney embedding theorem, fixing the dimension.
Realizing a graph as a subgraph of a product of linear
graphs is analogue to an isometric embedding of a Riemannian manifold in some Euclidean space
is a discrete Nash problem.  \\
C) If $f,g$ are two strongly locally injective function on some $d$-graph. Under which conditions
is it true that the level set $g=d$ in $f=c$ is topologically equivalent to the level set $g=d$ in $f=c$?
Random examples in a 3-sphere show that the answer is no in general. It might therefore be possible
that $f_2=0$ in $f_3=0$ is different that $f_3=0$ in $f_2=0$ but we have not yet an example, where
this difference takes place. \\
D) Assume $G$ is a $3$-sphere and the ground state $f_1$ (eigenvector to the smallest positive eigenvalue)
does not take the value $0$. Is the nodal hyper surface $\{f=0\}$ always a $2$-sphere? 
The case where $0$ is in the image of the ground state $f_1$ only appears in very rare cases. As long
as the $0$'s are isolated, we can randomly change the value on such places and not change
the topology of the nodal surface $\{ f = 0 \}$. 

\bibliographystyle{plain}

\end{document}